\documentclass[12pt]{article}
\usepackage{amsmath,amssymb,amsthm,amsfonts,setspace,graphicx,xcolor,mathabx,paralist,paralist,nicefrac,comment, hyperref, booktabs, appendix, titlesec, multirow, float, caption, thmtools,thm-restate,bbm,mathtools}

\usepackage{subcaption}
\usepackage[multiple]{footmisc}

\usepackage[style=authoryear-comp,natbib,maxbibnames=9,maxcitenames=3,uniquename = false, uniquelist=false,
    backend=biber]{biblatex}
\renewbibmacro{in:}{}

\usepackage{tikz}

\usepackage[margin=1in]{geometry}

\hypersetup{colorlinks=true,linkbordercolor=red,linkcolor=blue,citecolor=blue}

\usepackage{accents}
\newcommand{\ubar}[1]{\underaccent{\bar}{#1}}

\newtheorem{theorem}{Theorem}

\newtheorem{example}{Example}

\newtheorem{lemma}{Lemma}

\newtheorem{proposition}{Proposition}

\newtheorem*{cor-pur}{Corollary 4}
\newtheorem*{lem-tech}{Lemma 1}
\newtheorem*{lem-tech-DF}{Lemma 7}

\newtheorem*{assumption*}{Assumption}

\theoremstyle{definition}
\newtheorem{definition}{Definition}
\newtheorem*{definition2}{Definition 2'}

\newtheorem{remark}{Remark}

\renewcommand{\blacksquare}{\rule{0.5em}{0.5em}}
\renewenvironment{proof}[1][Proof]{\medskip\noindent\textbf{#1.} }{\ \hfill\blacksquare\\\medskip}

\title{Bayesian Polarization\footnote{I thank Ian Ball, Abhijit Banerjee, Roberto Corrao, Michael Crystal, Bnaya Dreyfuss, Drew Fudenberg, Nima Haghpanah, Stephen Morris, Shakked Noy, Anat Petruschka, Ariel Rubinstein, Ran Spiegler, Alexander Wolitzky, and Zvika Neeman for helpful comments and conversations. Refine.ink was used to check the paper for consistency and clarity.}}
\author{Tuval Danenberg\thanks{Department of Economics, MIT, tuvaldan@mit.edu}}

\date{\today}

\begin{document}

\maketitle
\thispagestyle{empty}
\vspace{-1em}

\begin{abstract}
   Discussions of political disagreement emphasize two patterns: \emph{polarization}, where beliefs diverge toward opposite extremes on each issue dimension; and \emph{issue alignment}, where individuals' views across issues become more internally consistent. We show that both can simultaneously arise under Bayesian learning from public information. We characterize the public signals that can induce persistent polarization on all dimensions and find that evidence of issue alignment can polarize Bayesian agents. However, we show that even stronger notions of polarization---requiring divergence beyond marginal beliefs---are inconsistent with Bayesian rationality. Whether multidimensional belief polarization translates into divergent aggregate positions depends on cross-issue separability. 
\end{abstract}

\newpage

\setcounter{page}{1}

\section{Introduction}

Polarization is widely viewed as a defining feature of contemporary politics. Merriam–Webster named polarization its 2024 Word of the Year, noting that it is ``one idea that both sides of the political spectrum agree on.” Consistent with this sentiment, a \textcite{pew2025} survey reports that 80\% of Americans believe that Republican and Democratic voters disagree not only over policies but also over basic facts.

The term polarization is often used as an umbrella for several related phenomena. A central notion is \textit{divergence}: policy preferences or ideological positions moving farther apart toward opposite extremes. A distinct but related notion is \textit{issue alignment}: individuals’ beliefs becoming more internally consistent, e.g., those who are liberal on one policy dimension are more likely to be liberal on others. Throughout the paper, we reserve the term polarization to refer specifically to divergence.

Focusing on political elites in the United States, \textcite{mccarty2019polarization} documents two stylized facts: elites---politicians at all levels, judges, and media figures---have become increasingly polarized since the 1970s, and this rise in polarization has coincided with increased issue alignment. For example, the correlation between legislators’ positions on economic issues and civil rights in the 1960s was substantially weaker than the correlation between positions on economic policy and affirmative action in the 1990s. Survey evidence points to similar patterns among the public. A \textcite{Gallup2023} study of attitudes on 24 key issues finds that partisan gaps have either persisted or widened across all dimensions, with opinions often moving in opposite directions \parencite[see also][]{Pew2014Polarization,pew2017partisan}. However, while there is broad agreement that strong partisans display polarization and issue alignment patterns similar to those of political elites, scholars debate whether the patterns observed in the general public reflect genuine belief change or only increased alignment between individuals’ views and their partisan identities.\footnote{\textcite{abramowitz2008polarization} argue that mass polarization is substantial, whereas \textcite{fiorina2008political} argue that it is largely a myth. See \textcite{lelkes2016mass} for a review.}

A popular explanation for these trends emphasizes selective exposure to information. Social media platforms are said to facilitate ``echo chambers,'' in which exposure to like-minded others reinforces prior beliefs \parencite{sunstein2018republic, levy2019echo}. But recent studies challenge this account: the extent of selective exposure online is limited \parencite{gentzkow2011ideological}, and some experimental studies find no link between exposure to online echo chambers and polarization \parencite{nyhan2023like}. Consistent with these findings, the \textcite{pew2025} survey reports that 67\% of Americans attribute partisan disagreement over basic facts primarily to differing interpretations of the same information.

Prima facie, this type of polarization appears irrational. Intuition as well as classic theoretical results suggest that rational Bayesian learning from the same information leads to belief convergence, so polarization should not arise. \textcite{baliga2013polarization} (henceforth BHK) present an impossibility result that formalizes this intuition. Considering two Bayesian agents with different priors, they show that observing a common public signal cannot cause beliefs to diverge in the stochastic dominance order.\footnote{Throughout the paper, stochastic dominance refers to first-order stochastic dominance; see formal definition below.}

In BHK, agents’ beliefs are defined over a one-dimensional state space. In practice, political beliefs are multidimensional, spanning foreign policy, the economy, cultural issues, and more. It may seem that if one-dimensional Bayesian polarization is impossible, then so is simultaneous polarization on multiple dimensions; in particular, patterns like those reported in the \textcite{Gallup2023} study are inconsistent with Bayesian updating from public information. Our first contribution is to overturn this intuition: when Bayesian agents' beliefs are defined over a multidimensional state space, a public signal can lead \textit{all} marginal beliefs to simultaneously diverge in the sense of stochastic dominance. Moreover, this polarization can persist even if agents observe an infinite sequence of i.i.d.\@ signals. Thus, the type of polarization that BHK rule out in one dimension after any public signal realization may arise coordinatewise across all dimensions, even in limit beliefs. 

Having established that multidimensional Bayesian polarization is possible, we turn to two natural questions: what type of information generates it, and what are its limits? Our second contribution is a characterization of the public signals that can generate persistent polarization of all marginal beliefs. The result suggests a novel informational mechanism consistent with the stylized facts: evidence of issue alignment can lead to polarization. Our third contribution is to delineate the limits of Bayesian polarization by establishing possibility and impossibility results for stronger notions of belief divergence. 

We establish these results within a simple model. Two Bayesian agents, $L$ and $H$, have prior beliefs over a finite multidimensional state space. There is a meaningful order across all dimensions. For example, the state may represent an optimal policy bundle, with dimensions representing economic, foreign, and climate policy, each ordered from the most liberal policy to the most conservative. Agent $L$ believes that the state is lower than agent $H$ on all dimensions, i.e., $L$ initially tends toward liberal policies. We will say that polarization occurs if there is a public signal that moves $L$'s beliefs further down on all dimensions and $H$'s beliefs further up.  Building on \textcite{dixit2007political}, who study polarization in a one-dimensional state space, we use stochastic dominance to compare marginal beliefs. So our notion of polarization requires that agent $L$'s posterior marginals are stochastically dominated by agent $L$'s prior marginals and vice versa for $H$. 

Theorem \ref{thm:cw possible} shows that this form of polarization is possible. The following example illustrates the result. 

\begin{example}\label{example:coordinatewise}
Consider a $2\times2$ state space $\Theta=\{x_1,x_2\}\times\{y_1,y_2\}$ with $x_1<x_2$ and $y_1<y_2$. For concreteness, interpret the $x$-dimension as economic policy and the $y$-dimension as immigration policy, with lower values corresponding to more liberal positions. Let $P^L$ and $P^H$ denote the prior beliefs of agents $L$ and $H$ over the optimal policy bundle, shown in Figure \ref{fig:example 1 priors}. Agent $L$ assigns more mass to low values on each dimension than agent $H$, so that every marginal of $P^L$ is stochastically dominated by the corresponding marginal of $P^H$.

 \begin{figure}[htbp]
  \centering
  \begin{minipage}[b]{0.48\textwidth}
    \centering
    \begin{tikzpicture}[scale=1.5]
      \draw[->, thick] (-0.3,0) -- (2.2,0) node[right] {$x$};
      \draw[->, thick] (0,-0.3) -- (0,2.2) node[above] {$y$};
      \draw[thick] (0,0) rectangle (2,2);
      \draw[thick] (1,0) -- (1,2);
      \draw[thick] (0,1) -- (2,1);
      \draw (0.5, -0.03) -- (0.5, 0.03) node[below] {$x_1$};
      \draw (1.5, -0.03) -- (1.5, 0.03) node[below] {$x_2$};
      \draw (-0.03, 0.5) -- (0.03, 0.5) node[left] {$y_1$};
      \draw (-0.03, 1.5) -- (0.03, 1.5) node[left] {$y_2$};
      \node at (0.5,0.5)   {\textbf{3/8}};
      \node at (1.5,0.5)   {\textbf{2/8}};
      \node at (0.5,1.5)   {\textbf{2/8}};
      \node at (1.5,1.5)   {\textbf{1/8}};
      \node[anchor=north east, font=\large\bfseries] at (0,0) {\(\mathbf{P^L}\)};
    \end{tikzpicture}
  \end{minipage}%
  \hspace{0.01\textwidth}%
  \begin{minipage}[b]{0.48\textwidth}
    \centering
    \begin{tikzpicture}[scale=1.5]
      \draw[->, thick] (-0.3,0) -- (2.2,0) node[right] {$x$};
      \draw[->, thick] (0,-0.3) -- (0,2.2) node[above] {$y$};
      \draw[thick] (0,0) rectangle (2,2);
      \draw[thick] (1,0) -- (1,2);
      \draw[thick] (0,1) -- (2,1);
      \draw (0.5, -0.03) -- (0.5, 0.03) node[below] {$x_1$};
      \draw (1.5, -0.03) -- (1.5, 0.03) node[below] {$x_2$};
      \draw (-0.03, 0.5) -- (0.03, 0.5) node[left] {$y_1$};
      \draw (-0.03, 1.5) -- (0.03, 1.5) node[left] {$y_2$};
      \node at (0.5,0.5)   {\textbf{1/8}};
      \node at (1.5,0.5)   {\textbf{2/8}};
      \node at (0.5,1.5)   {\textbf{2/8}};
      \node at (1.5,1.5)   {\textbf{3/8}};
      \node[anchor=north east, font=\large\bfseries] at (0,0) {\(\mathbf{P^H}\)};
    \end{tikzpicture}
  \end{minipage}
    \caption{Priors for Example \ref{example:coordinatewise}.}
  \label{fig:example 1 priors}
\end{figure}

Consider a public signal realization that reveals that the state lies on the ``diagonal'': $\{(x_1,y_1),(x_2,y_2)\}$, and is equally likely to be generated in either of these states. That is, the signal implies perfect positive association between the two dimensions while remaining uninformative about the value on each dimension. Let $Q^L$ and $Q^H$ denote the posterior beliefs obtained by Bayesian updating. These are shown in Figure \ref{fig:example 1 posteriors}.

 \begin{figure}[htbp]
  \centering
  \begin{minipage}[b]{0.48\textwidth}
    \centering
    \begin{tikzpicture}[scale=1.5]
      \draw[->, thick] (-0.3,0) -- (2.2,0) node[right] {$x$};
      \draw[->, thick] (0,-0.3) -- (0,2.2) node[above] {$y$};
      \draw[thick] (0,0) rectangle (2,2);
      \draw[thick] (1,0) -- (1,2);
      \draw[thick] (0,1) -- (2,1);
      \draw (0.5, -0.03) -- (0.5, 0.03) node[below] {$x_1$};
      \draw (1.5, -0.03) -- (1.5, 0.03) node[below] {$x_2$};
      \draw (-0.03, 0.5) -- (0.03, 0.5) node[left] {$y_1$};
      \draw (-0.03, 1.5) -- (0.03, 1.5) node[left] {$y_2$};
      \node at (0.5,0.5)   {\textbf{3/4}};
      \node at (1.5,0.5)   {\textbf{0}};
      \node at (0.5,1.5)   {\textbf{0}};
      \node at (1.5,1.5)   {\textbf{1/4}};
      \node[anchor=north east, font=\large\bfseries] at (0,0) {\(\mathbf{Q^L}\)};
    \end{tikzpicture}
  \end{minipage}%
  \hspace{0.01\textwidth}%
  \begin{minipage}[b]{0.48\textwidth}
    \centering
    \begin{tikzpicture}[scale=1.5]
      \draw[->, thick] (-0.3,0) -- (2.2,0) node[right] {$x$};
      \draw[->, thick] (0,-0.3) -- (0,2.2) node[above] {$y$};
      \draw[thick] (0,0) rectangle (2,2);
      \draw[thick] (1,0) -- (1,2);
      \draw[thick] (0,1) -- (2,1);
      \draw (0.5, -0.03) -- (0.5, 0.03) node[below] {$x_1$};
      \draw (1.5, -0.03) -- (1.5, 0.03) node[below] {$x_2$};
      \draw (-0.03, 0.5) -- (0.03, 0.5) node[left] {$y_1$};
      \draw (-0.03, 1.5) -- (0.03, 1.5) node[left] {$y_2$};
      \node at (0.5,0.5)   {\textbf{1/4}};
      \node at (1.5,0.5)   {\textbf{0}};
      \node at (0.5,1.5)   {\textbf{0}};
      \node at (1.5,1.5)   {\textbf{3/4}};
      \node[anchor=north east, font=\large\bfseries] at (0,0) {\(\mathbf{Q^H}\)};
    \end{tikzpicture}
  \end{minipage}
      \caption{Posteriors for Example \ref{example:coordinatewise}.}
  \label{fig:example 1 posteriors}
\end{figure}

Despite observing the same signal, the agents’ marginal beliefs diverge on each dimension. Writing $P_i^L,P_i^H$ and $Q_i^L,Q_i^H$ for the marginal distributions on dimension $i=1,2$, we have (here $\succeq_{st}$ is the stochastic dominance order)
\[
Q^L_i=(3/4,1/4)\prec_{st} P^L_i=(5/8,3/8)\prec_{st} P^H_i=(3/8,5/8)\prec_{st} Q^H_i=(1/4,3/4).
\]
Thus, agent $L$ becomes more confident that the state is low on each dimension, while agent $H$ becomes more confident that it is high, even though both update rationally from the same information.
\end{example}

Example \ref{example:coordinatewise}---which extends to any finite number of dimensions---shows that coordinatewise polarization is possible after a single public signal. As mentioned above, we also establish the possibility of \textit{limit} coordinatewise polarization---polarization in limit beliefs for an infinite sequence of i.i.d.\@ signals. 

Turning to the information that generates polarization, Theorem \ref{thm:cw conditions} presents a complete characterization of the signals that can lead to limit coordinatewise polarization in two-dimensional state spaces. An immediate necessary condition is \textit{partial identification}: there must be at least one additional state that is observationally equivalent to the true state of the world. Otherwise, beliefs converge to the truth, precluding polarization.\footnote{Formally, if the signal is $X$, the state space is $\Theta$, and the true state is $\hat\theta$, then there must be at least one state $\theta\in\Theta\setminus\{\hat\theta\}$ such that for every realization $x$: $\Pr[X=x|\theta]=\Pr[X=x|\hat\theta]$. Otherwise, Doob’s Consistency Theorem \parencite{doob1949application} implies that, from any full-support prior, beliefs converge to a point mass on $\hat\theta$.} Therefore, any polarizing signal must induce a set of states that are observationally equivalent to the true state. We call this the \emph{identified set}. The problem of characterizing limit polarizing signals reduces to characterizing the identified sets that permit polarization.

Under a mild technical condition, Theorem \ref{thm:cw conditions} establishes that an identified set can lead to coordinatewise polarization if and only if it satisfies three geometric conditions. Two of these—\textit{spanning} and \textit{balance}—govern how much the signal resolves disagreement about the value of the state. Spanning requires that both the identified set and its complement include states attaining the minimal and maximal values on each dimension, so that extreme outcomes remain possible in every direction. Balance rules out signals that decisively point toward uniformly high or uniformly low states. Taken together, these conditions ensure that the signal preserves uncertainty about the magnitude of the state, both within and across dimensions.

The third condition, \textit{non-compensation}, governs how the signal links the dimensions to one another. It rules out signals that systematically trade off one issue against another by pairing high values on one dimension with low values on the other. This condition is closely related to issue alignment: signals that imply perfect positive association between dimensions satisfy it, whereas signals implying perfect negative association violate it. Intuitively, polarization requires information that links the different issues without resolving disagreement about their values.

In the $2\times2$ state space, the theorem delivers a sharp implication. The diagonal \linebreak $\{(x_1,y_1),(x_2,y_2)\}$ is the \emph{unique} identified set that satisfies all three conditions. Consequently, the signal in Example \ref{example:coordinatewise} is essentially the only public signal that can generate persistent polarization on both dimensions. In this case, polarization necessarily coincides with increased issue alignment. In larger state spaces, increased alignment is no longer necessary for polarization; yet there always exist signals that generate both polarization and perfect alignment.

The result points to a mechanism through which public information about issue alignment can generate polarization. Information that reveals that dimensions are aligned---without revealing whether the state is overall high or low---can push beliefs toward opposite extremes. This information could be interpreted as a growing public understanding of an ideological structure that deems mixed positions---for example, liberal on economics and conservative on immigration---inconsistent \parencite{converse1964nature}. In line with the idea that agents can learn about alignment, \textcite{lupton2015political} show that such ideologically structured positions are positively associated with political sophistication. Regardless of the source of issue alignment, our contribution is to show that once it emerges, it can generate polarization through simple Bayesian updating.

Up to this point, we have considered polarization of all marginal beliefs, which we refer to as polarization in the \textit{coordinatewise order}. To clarify the constraints that Bayesian updating imposes on polarization, we next consider polarization in two stronger stochastic orders commonly used to compare multivariate distributions. Theorem \ref{thm:uo polarization} shows that polarization in the \textit{upper orthant order} is possible in the short-run but not in the long-run. Theorem \ref{thm:st impossible} shows that polarization in the \textit{stochastic order} is impossible following any single signal realization. Table \ref{tab:possible} summarizes these results. Here, we distinguish one-shot polarization, which occurs after a single signal realization, from limit polarization, which arises following an infinite sequence of i.i.d.\@ signals.

Taken together, the three orders define a feasibility frontier for Bayesian polarization. One-shot polarization is possible for any order weaker than the upper orthant order, but impossible for any order stronger than stochastic dominance. Limit polarization is possible for orders weaker than the coordinatewise order, but impossible for orders stronger than the upper orthant order.  

Mathematically, our result for the stochastic dominance order shows that the BHK impossibility result extends to higher dimensions. Substantively, however, it is the coordinatewise order---not stochastic dominance---that best reflects how polarization is measured in practice. Empirical work typically tracks divergence along individual issue dimensions such as climate policy, trade, or cultural issues. The coordinatewise possibility result shows that this form of polarization is fully consistent with Bayesian updating, even in the long-run limit. By contrast, the stronger orders capture divergence on \emph{combinations} of issues, e.g., whether the state is simultaneously high on several dimensions. The impossibility results imply that agreement may persist on such joint events, even when marginal beliefs diverge. Evidence of polarization on such issue combinations would be harder to reconcile with rational updating.

We conclude by studying when multidimensional belief polarization is reflected in one-dimensional aggregate positions. We show that the scope for a robust notion of polarization in positions depends on the separability of the aggregation rule across dimensions: additively separable aggregators allow persistent polarization, multiplicatively separable aggregators allow it only in the short run, and sufficiently strong cross-dimensional complementarities preclude it.

\begin{table}[t]
  \centering
  \caption{Possibility of Polarization}
  \label{tab:possible}
  \begin{tabular}{@{} l c c @{}}
    \toprule
    \multicolumn{1}{l}{\textbf{Stochastic Order}} 
      & \textbf{One‐shot Polarization} 
      & \textbf{Limit Polarization} \\
    \midrule
    Coordinatewise
      & $\checkmark$ 
      & $\checkmark$ \\
    \addlinespace
    Upper Orthant
      & $\checkmark$ 
      & X \\
    \addlinespace
    Stochastic Dominance    
      & X 
      & X \\
    \bottomrule
  \end{tabular}
\end{table}

\section{Related Literature}\label{sec:lit}

Our main notion of polarization builds on that of \textcite{dixit2007political}, who study Bayesian agents with beliefs over a one-dimensional state space and define polarization as divergence in the stochastic dominance order following a single public signal realization. They show that such polarization is impossible when signals satisfy the monotone likelihood ratio property (MLRP). BHK adopt this definition and prove that polarization is impossible for \emph{any} signal realization when agents are Bayesian and the state space is finite and one-dimensional.\footnote{\textcite{dixit2007political} present examples that violate MLRP, which they argue generate polarization. As noted by BHK, however, these examples exhibit divergence in expectations rather than divergence in stochastic dominance. \textcite{dixit2007political} also sketch examples of polarization in continuum state spaces; since BHK’s impossibility result applies only to finite state spaces, the possibility of polarization in the continuum remains open.
}

BHK's impossibility result reinforced the view that polarization is inconsistent with Bayesian learning. This aligns with classic results by \textcite{doob1949application} and \textcite{blackwell1962merging}, which establish distinct forms of belief convergence among Bayesians observing the same information. To reconcile the perceived theoretical impossibility with the apparent prevalence of polarization, a growing literature has introduced learning biases or other departures from standard Bayesian updating.\footnote{This literature is often motivated by same-evidence polarization experiments, starting with \textcite{lord1979biased}; see \textcite{rabin1999first} and \textcite{benoit2019apparent} for a review and discussion.} These include ambiguity aversion (BHK), confirmation bias \parencite{rabin1999first,fryer2019updating}, misperceptions about selective sharing on networks \parencite{bowen2023learning}, short-term memory \parencite{levy2025political}, and correlation neglect \parencite{levy2015does, loh2019dimensionality, ortoleva2015overconfidence}.

Other work has challenged this view by showing that some forms of polarization are consistent with Bayesian updating. In particular, several papers demonstrate that polarization along a single dimension may arise when agents hold beliefs over a multidimensional state space. \textcite{andreoni2012diverging} provide a simple theoretical example and experimental evidence, while \textcite{jern2014belief} present additional examples using Bayesian networks. A more general analysis is provided by \textcite{benoit2019apparent}, who characterize the public signals that can generate such polarization. Whereas they focus on beliefs after a single realization, \textcite{acemoglu2016fragility} consider agents observing an infinite sequence of i.i.d.\@ signals and establish that limit beliefs can diverge when agents are uncertain about signal likelihoods. In our terms, they demonstrate limit polarization on one dimension in a two-dimensional state space. However, their focus is on the \textit{magnitude} of disagreement rather than its \textit{direction}, showing that even vanishing uncertainty about likelihood functions can lead to substantial asymptotic disagreement. Taken together, these papers show that Bayesian learning does not preclude polarization, but suggest that for agents to polarize they need to be uncertain about some ancillary dimension, distinct from the dimension of disagreement.

We show that this restriction is unnecessary: Bayesian agents may simultaneously polarize on all dimensions of the state space. Moreover, on each dimension polarization occurs in the strong sense of stochastic dominance, and can persist in limit beliefs.\footnote{Most previous papers demonstrated polarization over dimensions with only two elements, where the choice of stochastic order is irrelevant.} In the short run, we find that an even stronger form of polarization---in the upper orthant order---is achievable.

Several papers have argued for the rationality of polarization outside the standard Bayesian paradigm or using different notions of polarization. \textcite{seidenfeld1993dilation} characterize belief dilation: the divergence of probability bounds for a family of beliefs following any signal realization. \textcite{nimark2019inattention} obtain persistent polarization due to rational inattention. \textcite{nielsen2021persistent} consider polarization of Bayesian beliefs on individual events and in the total variation distance. \textcite{dorst2023rational} argues that it can be rational for beliefs to diverge in ex-ante predictable directions.

Overall, existing work has reached seemingly conflicting conclusions on the compatibility of polarization with Bayesian learning, largely due to differing definitions. We contribute by clarifying \textit{which types} of polarization are feasible for Bayesian agents with beliefs over a finite multidimensional product space. Furthermore, while most previous papers focus either on short-run or limit polarization, we distinguish between the two and show that the distinction has bite: under the upper orthant order, polarization is possible only in the short run.

Finally, we analyze which public information can generate limit coordinatewise polarization, and show how evidence of issue alignment can itself be polarizing. \textcite{benoit2019apparent} also characterize polarizing signals, but consider a more specialized setting and a weaker notion of polarization.\footnote{In our terms, \textcite{benoit2019apparent} characterize the signals that lead to one-shot polarization on one dimension in a $2\times 2$ space, whereas we characterize the signals that lead to limit coordinatewise polarization in arbitrary finite two-dimensional spaces.} Their characterization serves as a stepping stone toward an analysis of \textit{population polarization}, which we do not consider. \textcite{spector2000rational} and \textcite{vaeth2025rational} present different models in which a form of rational learning generates both issue alignment and polarization: the former models agents learning about a common state from each other's cheap-talk messages, while the latter considers rationally inattentive agents learning about their idiosyncratic ideological positions. In contrast to these papers, we do not model why issue alignment arises. Instead, we observe that when it does, it can generate polarization through a Bayesian mechanism.

\section{Preliminaries}\label{sec:prelim}

\subsection{One-shot and Limit Polarization}
Two agents, $L$ and $H$, are learning about the value of a state $\hat\theta\in\Theta\subset \mathbb{R}^d$, where $d\ge2$. Each agent $i\in \{L,H\}$ has a prior $P^i$ with full support on $\Theta$. This state space is a finite product space, $\Theta=\Theta_1\times\cdots\times\Theta_d$, where $1<|\Theta_i|<\infty$ for all $i\in[d]$.\footnote{Throughout, for any integer $n\in\mathbb{N}$ we denote $[n]=\{1,2,\ldots,n\}$.}

Agents observe a realization $x\in \mathcal{X}$ of a public signal $X$, where $\mathcal{X}$ is finite. The likelihood function for this realization, known to both agents, is $\ell(\theta):=\Pr[X=x|\theta]$ for every state $\theta\in\Theta$.  Upon observing the signal, the agents follow Bayes rule to update their beliefs to the respective posteriors $Q^L$ and $Q^H$. We assume that $\ell(\theta)>0$ for at least one state, so that Bayesian updating is well defined.

\begin{definition}[One-shot polarization]\label{def:one shot polarization}
 Let $\succeq$ be a partial order over $\Delta(\Theta)$. Say that \emph{one-shot $\succeq$-polarization} occurs upon observing $x$ if
      \[
      Q^L\prec P^L\prec P^H\prec Q^H
      \]
  \end{definition}

We will also consider the case where the agents observe an infinite sequence of i.i.d.\@ public signals $X_1,X_2,\ldots$, all distributed as $X$. Here, we let $Q^L(t),Q^H(t)$ denote the posteriors after observing $X_1,\ldots,X_t$, and let $Q^L(\infty),Q^H(\infty)$ denote the limit posteriors. That is, $Q^L(\infty):=\lim_{t\rightarrow\infty}Q^L(t)$, and $Q^H(\infty):=\lim_{t\rightarrow\infty}Q^H(t)$ (we show below that there are unique limit distributions to which these sequences converge almost surely).  

\begin{definition}[Limit polarization]\label{def:limit polarization}
 Let $\succeq$ be a partial order over $\Delta(\Theta)$. Say that \emph{limit $\succeq$-polarization} occurs if the limit posteriors satisfy 
      \[
      Q^L(\infty)\prec P^L\prec P^H\prec Q^H(\infty)
      \]
  \end{definition}   

\subsection{Limit Beliefs and Partitional Signals}\label{subsec:limit}

This subsection presents an equivalent definition of limit polarization that simplifies the analysis. We begin with some general observations about limit beliefs. 

Consider an agent with prior $P\in\Delta(\Theta)$ observing an infinite sequence $X_1,X_2,\ldots $ of i.i.d.\@ signals supported on $\mathcal{X}$ (recall that $\Theta$ and $\mathcal{X}$ are finite). Let $\ell$ be the likelihood function of each of these signals, with $\ell(x;\theta)=\Pr[X=x|\theta]$.\footnote{Outside this section we fix the signal realization $x$ and simply write $\ell(\theta)$.} Let $\hat\theta$ be the true state of the world. The \emph{identified set} is
\[\hat\Gamma:=\{\,\theta\in\Theta \mid \ell(x;\theta)=\ell(x;\hat\theta)\,\,\forall x\in\mathcal{X}\,\}\]
Let $Q(t)$ be the posterior upon observing all signals up to and including $X_t$ and define $Q^*\in\Delta(\Theta)$ as:
\[Q^*(\theta) = \begin{cases}
 \frac{P(\theta)}{P(\hat\Gamma)}, & \theta\in\hat\Gamma \\
  0, & \text{else}
\end{cases}
\]
The following result is a simple corollary of Doob's Theorem on Bayesian consistency \parencite{doob1949application}.\footnote{A version of this corollary appears in Theorem 6.1 of \textcite{ramamoorthi2025doob}, for a more general setting. For completeness we provide a self-contained proof.}
\begin{lemma} \label{lemma:limit posterior}
   For every $\theta\in\Theta$, $Q(t)(\theta)\rightarrow Q^*(\theta)$ almost surely.
\end{lemma}

Lemma \ref{lemma:limit posterior} implies that the limit belief following an infinite sequence of signals is almost surely equal to the posterior belief that would arise upon observing a single signal that revealed that the state is in the identified set and nothing else. Formally, say that $X$ is a \emph{partitional signal} if there exists a partition $\mathcal{G}$ of $\Theta$ and a one-to-one mapping $x:\mathcal{G}\rightarrow \mathcal{X}$ such that for every $\theta\in \Gamma\in \mathcal{G}$: 
\[\ell(x(\Gamma);\theta) =1\text{ and } \ell(x';\theta) =0\text{ for all } x'\neq x(\Gamma). \]
The signal $X$ reveals what partition cell the state lies in, and two states in the same partition cell are observationally equivalent. 

Henceforth, we identify partitional signal realizations with their associated partition cells, i.e., we set $\mathcal{X}=\mathcal{G}$ and $x(\Gamma)=\Gamma$ for all $\Gamma\in\mathcal{G}$. Lemma \ref{lemma:limit posterior} implies that limit polarization is equivalent to one-shot polarization upon observing a partitional signal. That is, it shows that the following definition is equivalent to Definition \ref{def:limit polarization}.
\begin{definition2}
    Let $\succeq$ be a partial order over $\Delta(\Theta)$ and let $X$ be a partitional signal with partition $\mathcal{G}$. Let $P^L,P^H\in\Delta(\Theta)$ be the agents' priors and let $Q^L,Q^H$ be their respective posteriors upon observing a realization $\Gamma$ of $X$. Say that \emph{limit $\succeq$-polarization} occurs if 
     \[
      Q^L\prec P^L\prec P^H\prec Q^H
      \]
\end{definition2}
This formulation will be our working definition. One advantage is that it clarifies that limit polarization is a special case of one-shot polarization. An immediate implication is that if limit polarization is possible then so is one-shot polarization. Conversely, if one-shot polarization is impossible then limit polarization is impossible as well. 

\section{Coordinatewise Polarization}

\subsection{Definition}

Our definition of polarization requires a stochastic order to compare beliefs. Our primary notion of polarization will be polarization in the coordinatewise order, which we now define. Given two probability measures $P^L,P^H\in\Delta(\Theta)$, let $P^L_i,P^H_i$ be their respective marginals on the $i$ coordinate and $F^L_i, F^H_i$ be the cdfs of those marginals. We use $\succeq_{st}$ to denote the stochastic dominance order for univariate distributions. That is, we say $P^L_i$ is stochastically dominated by $P^H_i$, denoted $P^L_i\preceq_{st} P^H_i$, if $F^L_i(x)\ge F^H_i(x)$ for every $x\in\mathbb{R}$.
\begin{definition}[Coordinatewise dominance]
    Given two probability measures  $P^L,P^H\in\Delta(\Theta)$ say that $P^L$ is \emph{coordinatewise dominated} by $P^H$, denoted $P^L\preceq_{cw} P^H$, if $P^L_i\preceq_{st} P^H_i$ for all $i\in[d]$.
\end{definition}

Combining this definition of coordinatewise dominance with Definition \ref{def:one shot polarization}, we get that $\succeq_{cw}$-polarization (which we call coordinatewise polarization) occurs whenever $L$'s prior is stochastically dominated by $H$'s on every marginal, and a signal realization pushes all of $L$'s marginal beliefs further down and all of $H$'s marginal beliefs further up in the stochastic dominance order. If this holds after one signal realization then one-shot coordinatewise polarization occurs; if it holds for limit beliefs following an infinite sequence of signal realizations, then limit coordinatewise polarization occurs. Equivalently, one-dimensional polarization as defined by \textcite{dixit2007political} occurs for each marginal. 

As explained in the introduction, we take coordinatewise polarization as our main definition of polarization because it seems closest to the way multidimensional polarization is discussed and measured in practice. This crucially relies on the existence of a meaningful order---such as low-high, left-right or liberal-conservative---that applies to each dimension. When such an order exists, as it often does in practice, coordinatewise polarization can be interpreted as a case where the agents' beliefs diverge in this order following public information, e.g., agent $L$ becomes more liberal while agent $H$ becomes more conservative.

Our first result is that limit coordinatewise polarization is possible, which immediately implies that one-shot coordinatewise polarization is possible as well. Our second result characterizes the signals that lead to limit coordinatewise polarization. Before presenting these results, we address two possible critiques of coordinatewise polarization.

One could argue that the coordinatewise order is too restrictive for a definition of polarization. It requires divergence on \textit{all} dimensions and measures divergence on each dimension using stochastic dominance, itself a demanding criterion. Indeed, one might say that agents polarize if they diverge on most dimensions or if their marginal beliefs diverge in expectations. However, this critique only strengthens our possibility result: by establishing that polarization is possible in this strong sense, we immediately imply that polarization is possible according to any weaker order, such as divergence in means or on a subset of dimensions. The concern is nevertheless relevant for our characterization result---different characterizations would apply to weaker notions of divergence.

Conversely, one might argue that this order is too weak. Because it focuses on marginals, it allows for the possibility that beliefs converge on events that combine multiple dimensions. Moreover, standard multivariate stochastic orders used to compare distributions are strictly stronger than the coordinatewise order. Since discussions of multidimensional polarization typically consider beliefs on each separate issue (e.g., surveys ask respondents  about their economic positions and positions on national security separately) we view coordinatewise polarization as the most practically relevant measure. Nevertheless, it is theoretically important to understand whether stronger notions of polarization are consistent with Bayesian updating, a question we turn to in the next section. 

\subsection{Possibility of Coordinatewise Polarization}

\begin{theorem}\label{thm:cw possible}
  Limit $\succeq_{cw}$-polarization is possible. 
\end{theorem}

Example \ref{example:coordinatewise} in the introduction proves the theorem for the special case of a $2\times 2$ space. The general proof uses a similar construction: the identified set contains two points, $\ubar\theta$, the pointwise minimal state in $\Theta$, and $\bar\theta$, the pointwise maximal state. Agent $L$'s prior places more mass on $\ubar\theta$ than on $\bar\theta$ with the reverse holding for agent $H$, and both priors are uniform over all other states. Thus, we have $P^L\prec_{cw} P^H$. Furthermore, for sufficiently skewed priors, we find that revealing the state to be either $\ubar\theta$ or $\bar\theta$ pushes agent $L$'s beliefs down in the coordinatewise order and agent $H$'s beliefs up, i.e., coordinatewise polarization occurs. Since coordinatewise polarization occurs after a single realization of a  partitional signal, the discussion following Lemma \ref{lemma:limit posterior} implies that it can also arise after an infinite sequence of i.i.d.\@ signals so that limit  coordinatewise polarization is possible.\footnote{Formally, we have only defined one signal realization. But since this realization reveals that the state is in some set and nothing else, it is easy to define a partitional signal that has this realization in its support.} 

Theorem \ref{thm:cw possible} stands in stark contrast to the findings of BHK. While they find that for one-dimensional beliefs, polarization in stochastic dominance is impossible following any single signal realization, we find that when beliefs are multidimensional, polarization in stochastic dominance can arise simultaneously on all marginals. Thus, when beliefs are multidimensional, polarization on all individual issues is consistent with Bayesian updating. 

While the theorem shows that coordinatewise polarization is only possible, Appendix \ref{appendix:probability} demonstrates that it can also be \textit{probable}. We provide a simple example in which coordinatewise polarization occurs with high probability under both agents' priors. However, there is a tradeoff between the probability and magnitude of polarization.

\subsection{Which Signals Lead to Limit Coordinatewise Polarization}
 
The signal realization that polarizes agents in the proof of Theorem \ref{thm:cw possible} reveals that the state is perfectly aligned across the different dimensions and that it takes an extreme value on each dimension. Upon observing the realization, both agents believe that the state is either low on all dimensions or high on all dimensions; they just disagree on which of these is most likely. The signal generates both polarization and issue alignment: learning that the state is perfectly aligned leads the agents to polarize. 

Must polarization always coincide with issue alignment and extreme beliefs? More generally, what types of information can generate polarization? To answer these questions we now characterize the signal realizations that can lead to limit coordinatewise polarization, focusing on two-dimensional state spaces. Since limit beliefs are pinned down by the priors and the identified set, the problem reduces to specifying conditions on the identified set. We first review the standard partial orders on multidimensional Euclidean space and then define the conditions that will be used in Theorem \ref{thm:cw conditions}, the characterization theorem. 

 For two vectors $x,y\in\mathbb{R}^d$, say $x\le y$ if $x_i\le y_i$ for all $i\in[d]$, and say $x\ll y$ if $x_i<y_i$ for all $i\in[d]$. For a set $\Gamma\in\mathbb{R}^d$ define $\min(\Gamma)=\{\,x\in \Gamma \mid \nexists y\in\Gamma\setminus\{x\} \;\text{s.t.}\; y\le x\,\}
$. Similarly, $\max(\Gamma)$ is the subset of vectors in $\Gamma$ such that no other vector in $\Gamma$ is weakly above them. We will be considering finite sets for which the min and max sets are always nonempty. Since $\Theta$ is a finite product space, its min and max sets are singletons. Let $\ubar\theta$ denote the unique minimal element in $\Theta$ and $\bar\theta$ denote the unique maximal element. For general subsets $\Gamma\subset\Theta$, $\min(\Gamma)$ and $\max(\Gamma)$ may contain more than one element.

\begin{figure}[t]
\centering

\begin{subfigure}[t]{0.47\textwidth}
  \centering
  \begin{tikzpicture}[xscale=1.5, yscale=1.2]
    \path[use as bounding box] (-0.5,-0.5) rectangle (2.3,2.3);

    \draw[->, thick] (-0.3,0) -- (2.2,0) node[right] {$x$};
    \draw[->, thick] (0,-0.3) -- (0,2.2) node[above] {$y$};

    \foreach \x/\lab in {0.5/$x_1$,1.0/$x_2$,1.5/$x_3$}{
      \draw (\x,-0.03) -- (\x,0.03) node[below] {\lab};
    }
    \foreach \y/\lab in {0.5/$y_1$,1.0/$y_2$,1.5/$y_3$}{
      \draw (-0.03,\y) -- (0.03,\y) node[left] {\lab};
    }

    \foreach \xx in {0.5,1.0,1.5}{
      \foreach \yy in {0.5,1.0,1.5}{
        \node[draw, circle, fill=black, inner sep=1.8pt] at (\xx,\yy) {};
      }
    }

    \draw[red, thick] 
      (0.25, 1.25) -- (1.75, 1.25) --   
      (1.75, 0.75) -- (1.25, 0.75) --   
      (1.25, 0.25) -- (0.75, 0.25) --   
      (0.75, 0.75) -- (0.25, 0.75) --   
      cycle;

  \end{tikzpicture}
  \caption{$\Gamma$ not spanning}
\end{subfigure}
\hfill
\begin{subfigure}[t]{0.47\textwidth}
  \centering
  \begin{tikzpicture}[xscale=1.5, yscale=1.2]
    \path[use as bounding box] (-0.5,-0.5) rectangle (2.3,2.3);

    \draw[->, thick] (-0.3,0) -- (2.2,0) node[right] {$x$};
    \draw[->, thick] (0,-0.3) -- (0,2.2) node[above] {$y$};

    \foreach \x/\lab in {0.5/$x_1$,1.0/$x_2$,1.5/$x_3$}{
      \draw (\x,-0.03) -- (\x,0.03) node[below] {\lab};
    }
    \foreach \y/\lab in {0.5/$y_1$,1.0/$y_2$,1.5/$y_3$}{
      \draw (-0.03,\y) -- (0.03,\y) node[left] {\lab};
    }

    \foreach \xx in {0.5,1.0,1.5}{
      \foreach \yy in {0.5,1.0,1.5}{
        \node[draw, circle, fill=black, inner sep=1.8pt] at (\xx,\yy) {};
      }
    }

    \draw[red, thick] (0.25,0.25) rectangle (0.75,1.75);
    
    \draw[red, thick] (1.25,0.25) rectangle (1.75,0.75);

  \end{tikzpicture}
  \caption{$\Gamma^C$ not spanning}
\end{subfigure}

\par\bigskip

\begin{subfigure}[t]{0.47\textwidth}
  \centering
  \begin{tikzpicture}[xscale=1.5, yscale=1.2]
    \path[use as bounding box] (-0.5,-0.5) rectangle (2.3,2.3);
    \draw[->, thick] (-0.3,0) -- (2.2,0) node[right] {$x$};
    \draw[->, thick] (0,-0.3) -- (0,2.2) node[above] {$y$};
    \draw (0.5,-0.03) -- (0.5,0.03) node[below] {$x_1$};
    \draw (1.5,-0.03) -- (1.5,0.03) node[below] {$x_2$};
    \draw (-0.03,0.5) -- (0.03,0.5) node[left] {$y_1$};
    \draw (-0.03,1.5) -- (0.03,1.5) node[left] {$y_2$};
    \foreach \xx/\yy in {0.5/0.5, 1.5/0.5, 0.5/1.5, 1.5/1.5}{
      \node[draw, circle, fill=black, inner sep=1.8pt] at (\xx,\yy) {};
    }
    \begin{scope}[rotate around={45:(1,1)}]
      \draw[red, thick] (0.1,0.75) rectangle (1.9,1.25);
    \end{scope}
  \end{tikzpicture}
  \caption{Diagonal: $\Gamma$, $\Gamma^C$ spanning}
\end{subfigure}
\hfill
\begin{subfigure}[t]{0.47\textwidth}
  \centering
  \begin{tikzpicture}[xscale=1.5, yscale=1.2]
    \path[use as bounding box] (-0.5,-0.5) rectangle (2.3,2.3);
    \draw[->, thick] (-0.3,0) -- (2.2,0) node[right] {$x$};
    \draw[->, thick] (0,-0.3) -- (0,2.2) node[above] {$y$};
    \draw (0.5,-0.03) -- (0.5,0.03) node[below] {$x_1$};
    \draw (1.5,-0.03) -- (1.5,0.03) node[below] {$x_2$};
    \draw (-0.03,0.5) -- (0.03,0.5) node[left] {$y_1$};
    \draw (-0.03,1.5) -- (0.03,1.5) node[left] {$y_2$};
    \foreach \xx/\yy in {0.5/0.5, 1.5/0.5, 0.5/1.5, 1.5/1.5}{
      \node[draw, circle, fill=black, inner sep=1.8pt] at (\xx,\yy) {};
    }
    \begin{scope}[rotate around={-45:(1,1)}]
      \draw[red, thick] (0.1,0.75) rectangle (1.9,1.25);
    \end{scope}
  \end{tikzpicture}
  \caption{Off-diagonal: $\Gamma$, $\Gamma^C$ spanning.}
\end{subfigure}

\caption{Spanning conditions.}
\label{fig:spanning-conditions}
\end{figure}

\begin{definition}[Spanning]
 A set $\Gamma\subset\Theta$ is \emph{spanning} if for every coordinate $i$ there exist $\theta,\theta'\in\Gamma$ such that $\theta_i=\ubar\theta_i,\theta'_i=\bar\theta_i$.
 \end{definition}

  A set $\Gamma\subset\Theta$ is \emph{spanning} if it obtains the lowest and highest values on each coordinate. We will show that for a set to polarize, both the set and its complement must be spanning. Figure \ref{fig:spanning-conditions} demonstrates this condition. The top panels show violations: the top left is a set that is not spanning (does not attain the maximal value on the $y$-coordinate), and the top right is a set whose complement is not spanning (does not attain the minimal value on the $x$-coordinate). The bottom panels show the unique subsets of the $2\times 2$ state space that are both spanning and have a spanning complement; these are what we call the ``diagonal" and the ``off-diagonal".

The second condition in Theorem \ref{thm:cw conditions} is that the identified set be balanced.

 \begin{definition}[Biased and balanced] $\Gamma\subset\Theta$ is \emph{biased downward} if there exists $\theta\in\Theta$, with $\theta\gg\ubar\theta$, such that $\Gamma$ contains all states strictly below $\theta$  and none of the states weakly above it.\footnote{Formally, $\{\theta'\in\Theta | \theta'\ll\theta\}\subset\Gamma$ and $\{\theta'\in\Theta | \theta'\ge\theta\}\cap \Gamma=\emptyset.$} $\Gamma\subset\Theta$ is \emph{biased upward} if there exists $\theta\in\Theta$, with $\theta\ll \bar\theta$  such that $\Gamma$ contains all states strictly above $\theta$ and none of the states weakly below it. $\Gamma$ is \emph{balanced} if it is not biased downward or upward. 
 \end{definition}

\begin{figure}[t]
\centering

\begin{subfigure}[t]{0.47\textwidth}
  \centering
  \begin{tikzpicture}[xscale=1.5, yscale=1.2]
    \path[use as bounding box] (-0.5,-0.5) rectangle (2.3,2.3);

    \draw[->, thick] (-0.3,0) -- (2.2,0) node[right] {$x$};
    \draw[->, thick] (0,-0.3) -- (0,2.2) node[above] {$y$};

    \draw (0.5,-0.03) -- (0.5,0.03) node[below] {$x_1$};
    \draw (1.5,-0.03) -- (1.5,0.03) node[below] {$x_2$};
    \draw (-0.03,0.5) -- (0.03,0.5) node[left] {$y_1$};
    \draw (-0.03,1.5) -- (0.03,1.5) node[left] {$y_2$};

    \foreach \xx/\yy in {0.5/0.5, 1.5/0.5, 0.5/1.5, 1.5/1.5}{
      \node[draw, circle, fill=black, inner sep=1.8pt] at (\xx,\yy) {};
    }

    \draw[red, thick] 
      (0.3,0.3) -- (1.7,0.3) -- (1.7,0.7) -- (0.7,0.7)
      -- (0.7,1.7) -- (0.3,1.7) -- cycle;
  \end{tikzpicture}
  \caption{Biased downward}
\end{subfigure}
\hfill
\begin{subfigure}[t]{0.47\textwidth}
  \centering
  \begin{tikzpicture}[xscale=1.5, yscale=1.2]
    \path[use as bounding box] (-0.5,-0.5) rectangle (2.3,2.3);

    \draw[->, thick] (-0.3,0) -- (2.2,0) node[right] {$x$};
    \draw[->, thick] (0,-0.3) -- (0,2.2) node[above] {$y$};

    \draw (0.5,-0.03) -- (0.5,0.03) node[below] {$x_1$};
    \draw (1.5,-0.03) -- (1.5,0.03) node[below] {$x_2$};
    \draw (-0.03,0.5) -- (0.03,0.5) node[left] {$y_1$};
    \draw (-0.03,1.5) -- (0.03,1.5) node[left] {$y_2$};

    \foreach \xx/\yy in {0.5/0.5, 1.5/0.5, 0.5/1.5, 1.5/1.5}{
      \node[draw, circle, fill=black, inner sep=1.8pt] at (\xx,\yy) {};
    }

    \draw[red, thick] (0.3,1.3) rectangle (1.7,1.7);
  \end{tikzpicture}
  \caption{Biased upward}
\end{subfigure}

\par\bigskip

\begin{subfigure}[t]{0.47\textwidth}
  \centering
  \begin{tikzpicture}[xscale=1.5, yscale=1.2]
    \path[use as bounding box] (-0.5,-0.5) rectangle (2.3,2.3);

    \draw[->, thick] (-0.3,0) -- (2.2,0) node[right] {$x$};
    \draw[->, thick] (0,-0.3) -- (0,2.2) node[above] {$y$};

    \foreach \x/\lab in {0.5/$x_1$, 1.0/$x_2$, 1.5/$x_3$}{
      \draw (\x,-0.03) -- (\x,0.03) node[below] {\lab};
    }
    \foreach \y/\lab in {0.5/$y_1$, 1.0/$y_2$, 1.5/$y_3$}{
      \draw (-0.03,\y) -- (0.03,\y) node[left] {\lab};
    }

    \foreach \xx in {0.5,1.0,1.5}{
      \foreach \yy in {0.5,1.0,1.5}{
        \node[draw, circle, fill=black, inner sep=1.8pt] at (\xx,\yy) {};
      }
    }

    \draw[red, thick] 
      (0.2, 0.2) -- (1.25, 0.2) --  
      (1.25, 0.75) -- (1.8, 0.75) -- 
      (1.8, 1.25) -- (1.25, 1.25) -- 
      (1.25, 1.25) -- (0.2, 1.25) -- 
      cycle;

  \end{tikzpicture}
  \caption{Biased downward}
\end{subfigure}
\hfill
\begin{subfigure}[t]{0.47\textwidth}
  \centering
  \begin{tikzpicture}[xscale=1.5, yscale=1.2]
    \path[use as bounding box] (-0.5,-0.5) rectangle (2.3,2.3);

    \draw[->, thick] (-0.3,0) -- (2.2,0) node[right] {$x$};
    \draw[->, thick] (0,-0.3) -- (0,2.2) node[above] {$y$};

    \foreach \x/\lab in {0.5/$x_1$, 1.0/$x_2$, 1.5/$x_3$}{
      \draw (\x,-0.03) -- (\x,0.03) node[below] {\lab};
    }
    \foreach \y/\lab in {0.5/$y_1$, 1.0/$y_2$, 1.5/$y_3$}{
      \draw (-0.03,\y) -- (0.03,\y) node[left] {\lab};
    }

    \foreach \xx in {0.5,1.0,1.5}{
      \foreach \yy in {0.5,1.0,1.5}{
        \node[draw, circle, fill=black, inner sep=1.8pt] at (\xx,\yy) {};
      }
    }

    \draw[red, thick] (0.35, 1.35) rectangle (0.65, 1.65); 
    \draw[red, thick] (1.35, 1.35) rectangle (1.65, 1.65); 
    \draw[red, thick] (1.35, 0.35) rectangle (1.65, 0.65); 

  \end{tikzpicture}
  \caption{Biased upward}
\end{subfigure}

\caption{Violations of the balance condition.}
\label{fig:balanced-violations}
\end{figure}

Geometrically, a set is biased downward if it is contained in a ``lower L," while its complement is contained in the corresponding ``upper L." That is, there exists a state such that, for the Euclidean coordinate system originating at that state, $\Gamma$ contains all states in the lower left quadrant, while $\Gamma^C$ contains all states in the upper right quadrant. Similarly, a set is biased upward if it covers the upper right quadrant while its complement covers the lower left quadrant. Figure \ref{fig:balanced-violations} presents  sets that are not balanced in the $2\times 2$ and $3\times 3$ state spaces. For example, the set in panel (a) is biased downward because it contains all states strictly below $(x_2,y_2)$ and none of the states weakly above it. Note that in the $2\times 2$ state space, and generally in any $n\times n$ state space, both the diagonal and the off-diagonal (panels (c) and (d) of Figure \ref{fig:spanning-conditions}) are balanced.\footnote{In an $n\times n$ state space we define the diagonal as $\{(x_1,y_1),(x_2,y_2),\ldots, (x_n,y_n)\}$ and the off-diagonal as $\{(x_1,y_n),(x_2,y_{n-1}),\ldots,(x_n,y_1)\}$.} For any partition of the state space into quadrants, the diagonal  intersects both the upper-right and lower-left quadrants, while the off-diagonal does not cover either of them. 

The final condition in Theorem \ref{thm:cw conditions} is that the set be non-compensatory. We now define the negation of this condition. 

 \begin{definition}[Compensatory]
   $\Gamma\subset\Theta$ is \emph{compensatory} if there exists $\theta\in\Theta$, with $\theta\ll\bar\theta$ such that every state in $\Gamma$ is neither weakly below $\theta$ nor strictly above it. That is, for every $\theta'\in \Gamma$ neither $\theta'\le\theta$ nor $\theta'\gg\theta$ holds.\footnote{\label{fn:equivalent compensatory} Equivalently, there exists $\hat\theta\in\Theta$ with $\hat\theta\gg\ubar\theta$ such that for every $\theta'\in\Gamma$ neither $\theta'\ge \hat\theta$ nor $\theta'\ll\hat \theta$ holds. If the condition as it appears in the text is satisfied for $\theta\ll\bar\theta$ then the version in this footnote is satisfied for $\hat\theta=\min\{\theta'\in\Theta|\theta'\gg\theta\}$.}
 \end{definition}

\begin{figure}[t]
\centering

\begin{subfigure}[t]{0.47\textwidth}
  \centering
  \begin{tikzpicture}[xscale=1.5, yscale=1.2]
    \path[use as bounding box] (-0.5,-0.5) rectangle (2.3,2.3);

    \draw[->, thick] (-0.3,0) -- (2.2,0) node[right] {$x$};
    \draw[->, thick] (0,-0.3) -- (0,2.2) node[above] {$y$};

    \foreach \x/\lab in {0.5/$x_1$, 1.0/$x_2$, 1.5/$x_3$}{
      \draw (\x,-0.03) -- (\x,0.03) node[below] {\lab};
    }
    \foreach \y/\lab in {0.5/$y_1$, 1.0/$y_2$, 1.5/$y_3$}{
      \draw (-0.03,\y) -- (0.03,\y) node[left] {\lab};
    }

    \foreach \xx in {0.5,1.0,1.5}{
      \foreach \yy in {0.5,1.0,1.5}{
        \node[draw, circle, fill=black, inner sep=1.8pt] at (\xx,\yy) {};
      }
    }

    \begin{scope}[rotate around={45:(0.75,1.25)}]
      \draw[red, thick] (0.25,1.05) rectangle (1.25,1.45);
    \end{scope}

    \draw[red, thick] (1.35,0.35) rectangle (1.65,0.65);

  \end{tikzpicture}
  \caption{Compensatory}
\end{subfigure}
\hfill
\begin{subfigure}[t]{0.47\textwidth}
  \centering
  \begin{tikzpicture}[xscale=1.5, yscale=1.2]
    \path[use as bounding box] (-0.5,-0.5) rectangle (2.3,2.3);

    \draw[->, thick] (-0.3,0) -- (2.2,0) node[right] {$x$};
    \draw[->, thick] (0,-0.3) -- (0,2.2) node[above] {$y$};

    \foreach \x/\lab in {0.5/$x_1$, 1.0/$x_2$, 1.5/$x_3$}{
      \draw (\x,-0.03) -- (\x,0.03) node[below] {\lab};
    }
    \foreach \y/\lab in {0.5/$y_1$, 1.0/$y_2$, 1.5/$y_3$}{
      \draw (-0.03,\y) -- (0.03,\y) node[left] {\lab};
    }

    \foreach \xx in {0.5,1.0,1.5}{
      \foreach \yy in {0.5,1.0,1.5}{
        \node[draw, circle, fill=black, inner sep=1.8pt] at (\xx,\yy) {};
      }
    }

    \begin{scope}[rotate around={45:(0.75,1.25)}]
      \draw[red, thick] (0.25,1.05) rectangle (1.25,1.45);
    \end{scope}

    \draw[red, thick] (1.35,0.35) rectangle (1.65,0.65);

    \draw[red, thick] (0.35,0.35) rectangle (0.65,0.65);

  \end{tikzpicture}
  \caption{Non-compensatory}
\end{subfigure}

\caption{Compensatory and non-compensatory sets.}
\label{fig:compensatory}
\end{figure}

Geometrically, a set is compensatory if it is contained in off-diagonal quadrants, i.e., in the union of the upper-left and lower-right quadrants originating at some state. The left panel of Figure \ref{fig:compensatory} presents an example of such a set; the right panel presents an example of a non-compensatory set, as that set intersects the bottom-left quadrant. 

Finally, we introduce a strict form of dominance. For distributions supported on a finite ordered set $X$, let $P^L \prec^*_{st} P^H$ denote \emph{strong stochastic dominance}, defined by $F^L(x) > F^H(x)$ for all $\min (X) \le x < \max (X)$. We write $P^L \prec^*_{cw} P^H$ if $P^L_i \prec^*_{st} P^H_i$ for all dimensions $i$. We say that \emph{strong coordinatewise polarization} occurs if 
\[Q^L\prec_{cw} P^L\prec^*_{cw} P^H\prec_{cw} Q^H.
\]
We say that an identified set $\Gamma$ can strongly coordinatewise polarize if there exist priors for which strong coordinatewise polarization occurs for a partitional signal with realization $\Gamma$. Relative to coordinatewise polarization, the strong version rules out the knife-edge cases in which the marginal cdfs of $P^L$ and $P^H$ coincide at interior points of their supports.

We are now ready to state the characterization theorem. Here and below we let $\Gamma^C=\Theta\setminus\Gamma$ denote the complement of $\Gamma$ in $\Theta$. 

\begin{theorem}\label{thm:cw conditions}
    Let $\Theta\subset\mathbb{R}^2$. An identified set $\Gamma\subset\Theta$ can strongly coordinatewise polarize if and only if all of the following hold: (i) $\Gamma$ and $\Gamma^C$ are spanning, (ii) $\Gamma$ is balanced, (iii) $\Gamma$ is non-compensatory. 
\end{theorem}

The first step to proving Theorem \ref{thm:cw conditions} is rewriting the definition of polarization purely in terms of priors. Lemma \ref{lemma:Gamma and complement} shows that an agent's posterior belief moves down in the coordinatewise order upon learning that the identified set is $\Gamma$ if and only if their prior conditional on $\Gamma$ is coordinatewise dominated by their prior conditional on its complement.

\begin{lemma}\label{lemma:Gamma and complement}
    Let $P$ be a full support prior over $\Theta$ and $Q$ the posterior after observing a partitional signal with realization $\Gamma$. Let $P|_{\Gamma}, {P}|_{\Gamma^C}$ be the conditional distributions of $P$ over $\Gamma$ and $\Gamma^C$, respectively. Then $Q\preceq_{cw} P\iff P|_{\Gamma} \preceq_{cw}{P}|_{\Gamma^C}$. 
    \end{lemma}

The lemma implies that $\Gamma$ strongly coordinatewise polarizes given priors $P^L$ and $P^H$ if and only if the following three requirements are jointly satisfied: $P^L|_{\Gamma}\prec_{cw}P^L|_{\Gamma^C}$, $P^H|_{\Gamma}\succ_{cw} P^H|_{\Gamma^C}$, and $P^L\prec^*_{cw} P^H$. In words, the identified set must lead each agent to hold more extreme beliefs, relative to its complement, in the direction of their prior disagreement.   

The ``necessity direction'' of the proof of Theorem \ref{thm:cw conditions} establishes that if an identified set violates either the spanning or balance conditions, then individual updating is restricted: either the low agent's beliefs cannot move down in the coordinatewise order, or the high agent's beliefs cannot move up. For example, if $\Gamma$ is biased downward, then it cannot be the case that $P^H\prec_{cw}Q^H$. If an identified set satisfies these conditions but is compensatory, then---while it may be possible for some agent to update down on this identified set and for another agent to update up---it cannot be the case that the initially lower agent updates down while the high agent updates up simultaneously. This part of the proof uses the concept of generating functions, which are defined in Section \ref{sec:limits}.

Proving sufficiency is more involved. We show that if an identified set $\Gamma$ satisfies all three conditions, it is possible to construct priors $P^L$ and $P^H$ for which strong coordinatewise polarization occurs. The construction concentrates the mass of $P^L$ on $\min(\Gamma)$ and $\max(\Gamma^C)$ while concentrating the mass of $P^H$ on $\max(\Gamma)$ and $\min(\Gamma^C)$. Intuitively, the mass on $\min(\Gamma)$ versus $\max(\Gamma)$ ensures that $P^L\prec^*_{cw} P^H$ while the mass on $\max(\Gamma^C)$ versus $\min(\Gamma^C)$ ensures the correct updating directions: upon learning that the state is not in $\Gamma^C$, agent $L$ eliminates high states (moving beliefs down), while agent $H$ eliminates low states (moving beliefs up). The primary challenge lies in constructing priors such that the global requirement $P^L\prec^*_{cw}P^H$ does not conflict with the local structure on $\Gamma$ and $\Gamma^C$ needed to drive these opposing updates. A key step is Lemma \ref{lemma:ac dominance}, which is proved by induction on the size of the state space. Although the sufficiency proof does not provide the polarizing priors explicitly, the induction argument of Lemma \ref{lemma:ac dominance} can be adapted to construct them recursively. In general, the proof shows that many priors can be polarized by a given identified set.

We now make two technical remarks, and then discuss how the three conditions constrain the content of polarizing public information.

\begin{remark}[Strong coordinatewise polarization]
The fact that Theorem \ref{thm:cw conditions} applies to \textit{strong} coordinatewise polarization can be viewed as a genericity condition. The theorem characterizes the identified sets that can polarize priors $P^L$ and $P^H$, provided that the priors’ marginal cdfs do not coincide at interior points of their supports---a requirement that holds generically among full-support priors on $\Theta$. This strengthening of coordinatewise polarization is required only for the necessity direction, where it is used to establish the non-compensation condition.
\end{remark}

\begin{remark}[Higher dimensions]
While our theorem focuses on two-dimensional state spaces, the spanning condition remains necessary in higher dimensions: both the identified set and its complement must attain the extreme values along each dimension. We conjecture that suitable pairwise extensions of the balance and non-compensation conditions are necessary as well. However, our necessity arguments for these conditions, and especially our sufficiency proof, rely on the two-dimensional structure, so we do not have a complete characterization for $d>2$.
\end{remark}

\subsubsection*{Discussion}

The spanning and balance conditions restrict the information that agents receive about the \emph{magnitude} of the state, both within and across dimensions. Spanning requires that the signal preserve uncertainty about extreme values---on each dimension, both the minimal and maximal values must remain possible. At the same time, requiring that the complement of the identified set also be spanning implies that the signal must rule out some extreme states, thereby conveying information across dimensions: learning that the state is extreme on one dimension teaches the agents something about its value on another. The balance condition, however, limits the content of this information. While the signal may link dimensions, it cannot decisively indicate that the state is globally high or globally low.

The implication is that polarizing public information must allow for extreme positions on both sides of the disagreement while stopping short of resolving which side is correct. Signals that generate polarization must preserve the plausibility of very low and very high states on each dimension, yet avoid favoring one direction over another. This feature echoes journalistic norms of ``balanced coverage,'' which may appear counterintuitive given the common view that media bias is a primary driver of polarization. However, biased-media accounts typically operate through selective exposure, with different agents consuming different information that reinforces their prior beliefs \parencite{martin2017bias}. By contrast, we study polarization driven exclusively by \emph{public information}. In this setting, evidence that decisively favors one side cannot sustain polarization. Historical accounts suggest that balanced coverage can, in fact, contribute to disagreement. For instance, \textcite{oreskes2011merchants} document how the tobacco industry exploited the media’s adherence to the Fairness Doctrine to prolong controversy over the health risks of smoking, despite overwhelming scientific evidence.

The non-compensation condition rules out identified sets in which values above some threshold on one dimension are always paired with values below some threshold on the other. Geometrically, such sets are contained in the union of off-diagonal quadrants defined by that threshold, oriented from the top-left to the bottom-right of the state space. By contrast, non-compensatory sets cannot be contained in any union of such off-diagonal quadrants for any partition of the state space. In this sense, the condition rules out signals that reveal negative alignment between dimensions while allowing---though not requiring---that they reveal perfect positive alignment.\footnote{These quadrant-based notions of positive and negative alignment of multidimensional sets are related to stochastic orders of positive dependence of multivariate distributions, such as the positive quadrant order \parencite[e.g.,][]{shaked2007stochastic} and the recently proposed concentration along the diagonal order \parencite{basak2024similarity}. These orders, however, are typically defined for distributions with fixed marginals, and are therefore not well suited to our setting.}

In the $2\times2$ state space, the non-compensation condition implies that polarization necessarily coincides with increased issue alignment. Here, the diagonal and the off-diagonal are the only sets satisfying the spanning and balance conditions. Since the off-diagonal is compensatory, the diagonal is the unique polarizing set. While this exact coincidence is special to the $2\times2$ state space, the connection between polarization and issue alignment persists in larger spaces. In any two-dimensional state space, a subset that contains the extreme states $\ubar\theta$ and $\bar\theta$ is spanning, balanced, and non-compensatory. Thus, it can strongly coordinatewise polarize as long as its complement is spanning. In particular, in any $n\times n$ state space, the diagonal can polarize while the off-diagonal cannot. At the same time, there exist identified sets that satisfy all three conditions without increasing alignment. For example, in any two-dimensional state space larger than $2\times2$, the set $\Theta\setminus\{\ubar\theta,\bar\theta\}$ that contains all states \textit{except} the extremes satisfies all three conditions, yet reduces alignment in the sense that conditioning on this set shifts probability mass away from the diagonal.

As mentioned in the introduction, non-compensatory signals can arise when information reveals an underlying ideological structure linking different issues. Agents may learn that certain combinations of positions are internally inconsistent; for example, that one cannot be a Marxist when it comes to economics while endorsing a strictly traditionalist social ideology. Conditioning prior beliefs on the subset of states that are consistent with this ideological structure can generate polarization. While our presentation has focused on political beliefs, similar mechanisms may operate in other domains. In scientific contexts involving multiple related hypotheses, evidence about joint causes or correlations can likewise induce polarization by linking the hypotheses without resolving disagreement about their validity \parencite[see][for a discussion of scientific polarization]{o2018scientific}.

\section{Constraints on Bayesian Polarization}\label{sec:limits}

We have shown that a strong and practically relevant notion of polarization is fully consistent with Bayesian updating, even in the long-run limit. A natural question is whether Bayesian rationality imposes any constraints on polarization. To explore this, we consider two stricter notions: polarization in the multivariate stochastic dominance order and in the upper orthant order. We begin by defining the classes of sets that characterize these orders, as well as a third class that provides an alternative characterization of the coordinatewise order. 

A set $U\subset\Theta$ is an \emph{upper set} if $\theta\in U$ whenever $\theta\in\Theta$, $\theta\ge \delta$ and $\delta\in U$. A set $U\subset\Theta$ is an \emph{upper orthant} if it is of the form $\{\,\theta\in\Theta \mid \theta\gg a \,\}$ for some $a\in\mathbb{R}^d$. A set $U\subset\Theta$ is an \emph{upper projection} if it is of the form 
$\{\,\theta\in\Theta \mid \theta_i> a\,\}$ for some $a\in\mathbb{R}$ and $i\in[d]$.\footnote{Upper sets and upper orthants are standard in the stochastic orders literature; upper projections are not. Equivalently, these classes of sets can be defined over $\mathbb{R}^d$, in which case we would say that a set $U\subset \Theta$ is an upper set if there exists an upper set $\tilde U\subset\mathbb{R}^d$ such that $U=\Theta\cap\tilde U$.} 

\begin{definition}[Stochastic orders]\label{def:stochastic orders}
     Let $P^L,P^H$ be two probability measures over $\Theta$ and $P^L_i,P^H_i$ be their respective marginals on the $i$ coordinate. We say that:
     \begin{enumerate}
        \item $P^L$ is \emph{stochastically dominated} by $P^H$, denoted $P^L\preceq_{st} P^H$ if $P^L( U)\le P^H(U)$ for all upper sets $U\subset\Theta$.
        \item $P^L$ is \emph{upper orthant dominated} by $P^H$, denoted $P^L\preceq_{uo} P^H$ if $P^L( U)\le P^H(U)$ for all upper orthants $U\subset\Theta$.
        \item $P^L$ is \emph{coordinatewise dominated} by $P^H$, denoted $P^L\preceq_{cw} P^H$ if $P^L( U)\le P^H(U)$ for all upper projections $U\subset\Theta$. Equivalently, $P^L_i\preceq_{st} P^H_i$ for all $i\in[d]$.
     \end{enumerate}
\end{definition}

 Stochastic dominance and upper orthant dominance are the two standard generalizations of univariate stochastic dominance to multiple dimensions, with the former being the most common \parencite{shaked2007stochastic}.

It is immediate that every upper projection is an upper orthant, but not vice versa, and that every upper orthant is an upper set, but not vice versa. This proves the following result, which states that the orders in Definition \ref{def:stochastic orders} are ordered from the strongest, or most restrictive, to the weakest.\footnote{This relationship is also noted by \textcite{shaked2007stochastic}, who do not formally define coordinatewise dominance but mention that upper orthant dominance implies univariate stochastic dominance on each coordinate.} 
\begin{lemma} $P^L\preceq_{st} P^H\Longrightarrow  P^L\preceq_{uo} P^H \Longrightarrow  P^L\preceq_{cw} P^H$
    and none of the reverse implications hold. 
\end{lemma}

The stochastic orders defined above belong to the class of \emph{integral stochastic orders} \parencite{whitt1986stochastic, muller1997stochastic}. An order $\succeq$ in this class is characterized by a \textit{generator} $\mathcal{U}$, which is a set of functions such that $P^L\preceq P^H$ if and only if $E^L[u]\le E^H[u]$ for all $u\in\mathcal{U}$. Here and below, $E^L[u]$ and $E^H[u]$ are the expectations of $u$ according to the distributions $P^L$ and $P^H$, respectively. Table \ref{tab:generating functions} summarizes the generators for the three orders. By an increasing function, we mean a function $u:\Theta\rightarrow\mathbb{R}$ that is increasing according to the order $\ge$. A product of nonnegative univariate increasing functions is a function $u:\Theta\rightarrow\mathbb{R}$ such that $u(\theta)=\prod_{i=1}^d u_i(\theta_i)$ where each $u_i$ is nonnegative and increasing on $\Theta_i$, and a sum of univariate increasing functions is defined similarly. The characterizations for $\succeq_{st}$ and $\succeq_{uo}$ are standard \parencite{shaked2007stochastic}, while the characterization for $\succeq_{cw}$ is derived in Lemma \ref{lemma:cw generating}.\footnote{By the linearity of expectation, if an order is generated by $\mathcal{U}$, it is also generated by the convex cone spanned by $\mathcal{U}$.}

\begin{table}[ht]
  \centering
  \caption{Generating Functions}
  \label{tab:generating functions}
  \begin{tabular}{@{}ll@{}}
    \toprule
    \textbf{Stochastic Order} & \textbf{Generator} \\
    \midrule
    \addlinespace[5pt]
    $\succeq_{st}$ & Increasing functions \\
    \addlinespace[5pt]
    $\succeq_{uo}$ & Products of nonnegative \\
                   & univariate increasing functions \\
    \addlinespace[5pt]
    $\succeq_{cw}$ & Sums of univariate increasing functions \\
    \addlinespace[3pt]
    \bottomrule
  \end{tabular}
\end{table}

\begin{lemma}\label{lemma:cw generating}
   The order $\succeq_{cw}$ is generated by the set \[\mathcal{U}=\{\, u\in\mathbb{R}^{\Theta}\mid u(\theta)=\sum_{i=1}^d u_i(\theta_i)\text{ where } u_i:\Theta_i\rightarrow\mathbb{R} \text{ are increasing}\,\}. \]
\end{lemma}

We now turn to polarization in the upper orthant order. The proof of Theorem \ref{thm:uo polarization} can easily be modified to apply to the closely related lower orthant order.
\begin{theorem}\label{thm:uo polarization}
 One-shot $\succeq_{uo}$-polarization is possible. Limit $\succeq_{uo}$-polarization is impossible. 
\end{theorem}

To understand the impossibility result, recall that limit beliefs concentrate on the identified set---the mass on all states in the identified set increases while the mass on all states outside the set vanishes. Consider a $2\times 2$ state space as in Example \ref{example:coordinatewise}. This state space has four upper orthants, illustrated in Figure \ref{fig:upper-orthants}. Any identified set $\Gamma$ that contains $\bar\theta=(x_2,y_2)$ cannot lead to limit upper orthant polarization, because for both agents the mass on the singleton upper orthant $U_4=\{\bar\theta\}$ increases after observing the signal. In this case, $Q^L$ is not upper orthant dominated by $P^L$, so upper orthant polarization does not occur. On the other hand, if the identified set does not include $\bar\theta$, then for both agents the mass on $U_4$ decreases. This implies that $P^H$ is not upper orthant dominated by $Q^H$, again precluding polarization. This argument extends to any state space $\Theta$, establishing the impossibility of limit upper orthant polarization.

\begin{figure}[t]
  \centering
  
  \begin{subfigure}[t]{0.47\textwidth}
    \centering
    \begin{tikzpicture}[xscale=1.5, yscale=1.2]
      \path[use as bounding box] (-0.5,-0.5) rectangle (2.3,2.3);

      \draw[->, thick] (-0.3,0) -- (2.2,0) node[right] {$x$};
      \draw[->, thick] (0,-0.3) -- (0,2.2) node[above] {$y$};
      
      \draw (0.5, -0.03) -- (0.5, 0.03) node[below] {$x_1$};
      \draw (1.5, -0.03) -- (1.5, 0.03) node[below] {$x_2$};
      \draw (-0.03, 0.5) -- (0.03, 0.5) node[left] {$y_1$};
      \draw (-0.03, 1.5) -- (0.03, 1.5) node[left] {$y_2$};
  
      \foreach \xx in {0.5,1.5}{
        \foreach \yy in {0.5,1.5}{
          \node[draw, circle, fill=black, inner sep=1.8pt] at (\xx,\yy) {};
        }
      }
      
      \draw[red, thick] (0.3,0.3) rectangle (1.7,1.7);
    \end{tikzpicture}
    \caption{$U_1$}
  \end{subfigure}
  \hfill
  \begin{subfigure}[t]{0.47\textwidth}
    \centering
    \begin{tikzpicture}[xscale=1.5, yscale=1.2]
      \path[use as bounding box] (-0.5,-0.5) rectangle (2.3,2.3);

      \draw[->, thick] (-0.3,0) -- (2.2,0) node[right] {$x$};
      \draw[->, thick] (0,-0.3) -- (0,2.2) node[above] {$y$};
      
      \draw (0.5, -0.03) -- (0.5, 0.03) node[below] {$x_1$};
      \draw (1.5, -0.03) -- (1.5, 0.03) node[below] {$x_2$};
      \draw (-0.03, 0.5) -- (0.03, 0.5) node[left] {$y_1$};
      \draw (-0.03, 1.5) -- (0.03, 1.5) node[left] {$y_2$};
      
      \foreach \xx in {0.5,1.5}{
        \foreach \yy in {0.5,1.5}{
          \node[draw, circle, fill=black, inner sep=1.8pt] at (\xx,\yy) {};
        }
      }

      \draw[red, thick] (0.3,1.3) rectangle (1.7,1.7);
    \end{tikzpicture}
    \caption{$U_2$}
  \end{subfigure}

  \par\bigskip

  \begin{subfigure}[t]{0.47\textwidth}
    \centering
    \begin{tikzpicture}[xscale=1.5, yscale=1.2]
      \path[use as bounding box] (-0.5,-0.5) rectangle (2.3,2.3);

      \draw[->, thick] (-0.3,0) -- (2.2,0) node[right] {$x$};
      \draw[->, thick] (0,-0.3) -- (0,2.2) node[above] {$y$};
      
      \draw (0.5, -0.03) -- (0.5, 0.03) node[below] {$x_1$};
      \draw (1.5, -0.03) -- (1.5, 0.03) node[below] {$x_2$};
      \draw (-0.03, 0.5) -- (0.03, 0.5) node[left] {$y_1$};
      \draw (-0.03, 1.5) -- (0.03, 1.5) node[left] {$y_2$};
      
      \foreach \xx in {0.5,1.5}{
        \foreach \yy in {0.5,1.5}{
          \node[draw, circle, fill=black, inner sep=1.8pt] at (\xx,\yy) {};
        }
      }

      \draw[red, thick] (1.3,0.3) rectangle (1.7,1.7);
    \end{tikzpicture}
    \caption{$U_3$}
  \end{subfigure}
  \hfill
  \begin{subfigure}[t]{0.47\textwidth}
    \centering
    \begin{tikzpicture}[xscale=1.5, yscale=1.2]
      \path[use as bounding box] (-0.5,-0.5) rectangle (2.3,2.3);

      \draw[->, thick] (-0.3,0) -- (2.2,0) node[right] {$x$};
      \draw[->, thick] (0,-0.3) -- (0,2.2) node[above] {$y$};
      
      \draw (0.5, -0.03) -- (0.5, 0.03) node[below] {$x_1$};
      \draw (1.5, -0.03) -- (1.5, 0.03) node[below] {$x_2$};
      \draw (-0.03, 0.5) -- (0.03, 0.5) node[left] {$y_1$};
      \draw (-0.03, 1.5) -- (0.03, 1.5) node[left] {$y_2$};
      
      \foreach \xx in {0.5,1.5}{
        \foreach \yy in {0.5,1.5}{
          \node[draw, circle, fill=black, inner sep=1.8pt] at (\xx,\yy) {};
        }
      }

      \draw[red, thick] (1.3,1.3) rectangle (1.7,1.7);
    \end{tikzpicture}
    \caption{$U_4$}
  \end{subfigure}

  \caption{The four upper orthants in a $2 \times 2$ state space.}
  \label{fig:upper-orthants}
\end{figure}

With partitional signals, agents must update in the same direction on any singleton event. By contrast, general signals allow beliefs to move in opposite directions even on singletons. We show that this allows for one-shot upper orthant polarization. The proof builds on the observation that an agent's belief on a state increases if and only if the state's likelihood exceeds the expected likelihood according to the agent's prior. This is a special case of the following simple result. 

\begin{lemma}\label{lemma:posterior increase}
    Let $P\in\Delta(\Theta)$ be a prior belief with full support over $\Theta$ and $Q\in\Delta(\Theta)$ be the corresponding posterior after observing a realization $x$ of a signal $X$ with likelihood $\ell(\theta)=\Pr[X=x\mid \theta]$. Then, for every nonempty $A\subset \Theta$: $Q(A)>P(A)$ if and only if $E_P[\ell|A]>E_P[\ell]$.
\end{lemma}

 To see why one-shot upper orthant polarization is possible, consider the $2\times 2$ space and a signal realization $x$ with the likelihood function $\ell$ presented in Figure \ref{fig:likelihood uo}. If $L$'s prior mostly concentrates on $\ubar\theta$ (where $\ell=1$), their expected likelihood is close to $1$, causing their belief on $\bar\theta$ (where $\ell=1/2$) to fall after observing $x$. Conversely, we can construct $H$'s prior so that their expected likelihood is below $1/2$, so that their belief on $\bar\theta$ increases after observing $x$. With suitable priors, this divergence on the singleton upper orthant extends to all upper orthants. Note that unlike the symmetric signals that generate limit coordinatewise polarization, here the likelihoods of $\ubar\theta$ and $\bar\theta$ differ; this asymmetry enables one-shot polarization but also identifies the state. Thus, if the agents observe an infinite sequence of i.i.d.\@ signals, their beliefs will eventually concentrate on the true state.

\begin{figure}[t]
  \centering
  \begin{minipage}[b]{0.48\textwidth}
    \centering
    \begin{tikzpicture}[scale=1.5]
      \draw[->, thick] (-0.3,0) -- (2.2,0) node[right] {$x$};
      \draw[->, thick] (0,-0.3) -- (0,2.2) node[above] {$y$};
      \draw[thick] (0,0) rectangle (2,2);
      \draw[thick] (1,0) -- (1,2);
      \draw[thick] (0,1) -- (2,1);
      \draw (0.5,-0.03) -- (0.5,0.03) node[below] {$x_1$};
      \draw (1.5,-0.03) -- (1.5,0.03) node[below] {$x_2$};
      \draw (-0.03,0.5) -- (0.03,0.5) node[left] {$y_1$};
      \draw (-0.03,1.5) -- (0.03,1.5) node[left] {$y_2$};
      \node at (0.5,0.5)  {\textbf{1}};
      \node at (1.5,0.5)  {\textbf{0}};
      \node at (0.5,1.5)  {\textbf{0}};
      \node at (1.5,1.5)  {\textbf{1/2}};
      \node[anchor=north east, font=\large\bfseries] at (0,0) {$\ell(\theta)$};
    \end{tikzpicture}
  \end{minipage}
   \caption{Likelihood function for upper orthant polarization.}
  \label{fig:likelihood uo}
\end{figure} 

Since the upper orthant order is weaker than stochastic dominance, the impossibility of limit upper orthant polarization implies that limit stochastic dominance polarization is impossible as well. The following result then completes Table \ref{tab:possible}. 

\begin{theorem}\label{thm:st impossible}
 One-shot $\succeq_{st}$-polarization is impossible. 
\end{theorem}

Theorem \ref{thm:st impossible} extends Theorem 1 of BHK from finite one-dimensional state spaces to arbitrary finite multidimensional state spaces.\footnote{The proof does not rely on the assumption that $\Theta$ is a product space.} We first briefly review their argument. Consider a finite one-dimensional state space $\theta_1<\theta_2<\ldots<\theta_n$. For polarization in stochastic dominance to occur, the low agent's beliefs must increase on the lowest state while the high agent's beliefs decrease, and the opposite must occur on the highest state.\footnote{That is, $Q^L(\theta_1)\ge P^L(\theta_1)$, $Q^H(\theta_1)\le P^H(\theta_1)$, while $Q^L(\theta_n)\le P^L(\theta_n)$, $Q^H(\theta_n)\ge P^H(\theta_n)$.} BHK show that if these occur simultaneously, then beliefs at these states do not move at all (we extend this argument in Proposition \ref{prop:abstract polarization} below). 
They then proceed by induction to show that beliefs are constant on every state---i.e., $Q^L(\theta_i)=P^L(\theta_i)$ and $P^H(\theta_i)=Q^H(\theta_i)$ for every $i\in[n]$---so polarization does not occur. 

Our proof adapts this strategy to multidimensional state spaces. Instead of inducting on states, we induct on the size of upper sets. That is, we show that for every upper set $U$, $Q^H(U)=P^H(U)$ and $Q^L(U)=P^L(U)$ by induction on the size of $U$. The key property is that if $U$ is an upper set of size $m$, then removing any of its minimal elements leads to an upper set of size $m-1$.

To understand why one-shot upper orthant polarization is possible while one-shot stochastic polarization is not, consider again the $2\times 2$ state space. Here, there is one upper set that is not an upper orthant: the ``upper L" presented in Figure \ref{fig:add} (recall that all upper orthants are upper sets, but not vice versa). For $\succeq_{st}$-polarization to occur, the high agent's beliefs must increase on the upper set $\{\bar\theta\}$ and decrease on $\{\ubar\theta\}$, which is the complement of the upper L. The opposite must occur for the low agent. However, this pattern of divergence violates a general restriction implied by Bayesian updating, which we formalize in the following proposition.

\begin{figure}[t]
  \centering
  \begin{tikzpicture}[xscale=1.5, yscale=1.2]
    \draw[->, thick] (-0.3,0) -- (2.2,0) node[right] {$x$};
    \draw[->, thick] (0,-0.3) -- (0,2.2) node[above] {$y$};
    \draw (0.5,-0.03) -- (0.5,0.03) node[below] {$x_1$};
    \draw (1.5,-0.03) -- (1.5,0.03) node[below] {$x_2$};
    \draw (-0.03,0.5) -- (0.03,0.5) node[left] {$y_1$};
    \draw (-0.03,1.5) -- (0.03,1.5) node[left] {$y_2$};

    \node[draw, circle, fill=black, inner sep=2pt] at (0.5,0.5) {};
    \node[draw, circle, fill=black, inner sep=2pt] at (1.5,0.5) {};
    \node[draw, circle, fill=black, inner sep=2pt] at (0.5,1.5) {};
    \node[draw, circle, fill=black, inner sep=2pt] at (1.5,1.5) {};

    \draw[red, thick]
      (1.7,1.7) -- (0.3,1.7) -- (0.3,1.3) -- (1.3,1.3)
      -- (1.3,0.3) -- (1.7,0.3) -- cycle;
  \end{tikzpicture}
  \caption{Upper set that is not an upper orthant.}
  \label{fig:add}
\end{figure}

\begin{proposition}
    \label{prop:abstract polarization}
    Let $P, P'\in\Delta(\Theta)$ be full support priors over a finite state space $\Theta$ and $Q,Q'$ be the corresponding posteriors following a signal realization $x$. Assume that the likelihood function $\ell(\theta)=\Pr[X=x|\theta]$ is not constant in $\theta$. Then:
    \begin{enumerate}
        \item There exist two states---the minimum and maximum likelihood states---on which beliefs move in the same direction, i.e., $\left(Q(\theta)-P(\theta)\right)\left(Q'(\theta)-P'(\theta)\right)>0$. On all other states, beliefs may move in opposite directions. 
        \item  If beliefs move in opposite directions, the divergence is consistent across states: if there exists a state $\theta\in\Theta$ for which $Q(\theta)-P(\theta)\le 0\le Q'(\theta)-P'(\theta)$, with at least one of these inequalities strict, then there does not exist a state $\theta'$ at which $Q(\theta')-P(\theta')\ge 0\ge Q'(\theta')-P'(\theta')$.
    \end{enumerate}
\end{proposition}

Part 2 of Proposition \ref{prop:abstract polarization} formalizes the ``consistency of direction" constraint suggested by the BHK proof. It explains why the upper L creates an impossibility: stochastic polarization requires agents to reverse their direction of disagreement between the minimal and maximal states. In contrast, upper orthant polarization does not require divergence on the upper L set, so while agents' beliefs do diverge on at least one individual state ($\bar\theta$), the direction of divergence is consistent across states. 

Part 1 of Proposition \ref{prop:abstract polarization} highlights that Bayesian updating otherwise imposes only mild constraints on \textit{where} beliefs can diverge, allowing them to move in opposite directions on all but two states.\footnote{Note that updating in opposite directions is not synonymous with polarization, as agents may update in opposite directions \textit{towards each other}. But it is straightforward to construct an example where agents' beliefs move \textit{away} from each other on all but two states.} This distinction helps clarify the mechanics of our earlier results. Polarization operates through two channels. The first is \textit{state-level polarization}: beliefs moving in opposite directions on individual states. The second is \textit{event-level polarization}: beliefs moving in the same direction at the state level, but accumulating differently on events due to differences in the magnitude of updates. In the short run, beliefs can diverge on many states but in the long-run limit they must move in the same direction on every state. Hence, one-shot polarization operates through both channels while limit polarization operates purely through the second channel.

\section{Robust Polarization in Aggregate Positions}

So far we have focused on polarization of multidimensional beliefs, analyzing notions that require beliefs to diverge on rich sets of events. Yet, in several applications of interest, agents' multidimensional beliefs are mapped into one-dimensional aggregate positions, and polarization is discussed in terms of these positions. For instance, if beliefs are over the optimal policy bundle on several issue dimensions as in Example \ref{example:coordinatewise}, with policies on each dimension ordered from most liberal to most conservative, one could also discuss polarization in terms of agents' overall ideological positions. This aggregation of beliefs into positions could be performed by the agents themselves or by an analyst. For example, the DW-NOMINATE scaling approach uses legislators’ roll-call voting records to estimate low-dimensional ideal points. These estimates are often interpreted as ideological positions and have played a central role in studies of polarization in the U.S. Congress \parencite{mccarty2016polarized}.

To formalize this type of aggregation within our model, we fix an \emph{aggregator} $u:\Theta\rightarrow \mathbb{R}$ and define the agents' prior \emph{positions} to be $E_{P^L}[u]$ and $E_{P^H}[u]$. That is, an agent's prior position is simply the expectation of $u$ under their prior. Likewise, we define their posterior positions as $E_{Q^L}[u]$ and $E_{Q^H}[u]$. We define two notions of polarization in positions upon observing a realization $x$ of the signal $X$.

\begin{definition}\label{def:position polarization}
  Given an aggregator $u:\Theta\rightarrow\mathbb{R}$, we say that $x$ leads to \textit{one-shot polarization in positions} if 
\begin{equation}\label{eqn:position polarization}
    E_{Q^L}[u]\le E_{P^L}[u]\le E_{P^H}[u]\le E_{Q^H}[u]
\end{equation}
  Given a set of aggregators $\mathcal{U}$, we say that $x$ leads to \textit{$\mathcal{U}$-robust one-shot polarization in positions} if it leads to one-shot polarization in positions for all $u\in \mathcal{U}$, and there is a $u\in\mathcal{U}$ for which all inequalities in \eqref{eqn:position polarization} are strict.
  \end{definition}

  Limit polarization in positions and $\mathcal{U}$-robust limit polarization in positions are defined analogously, replacing beliefs after the realization $x$ with limit beliefs following a sequence of i.i.d.\@ signals (or, equivalently, requiring that $X$ be a partitional signal).

Definition \ref{def:position polarization} considers two agents with differing beliefs whose positions are aggregated using the same aggregator $u$. This captures settings in which agents may disagree on issues but share a common convention for translating multidimensional beliefs into a one-dimensional position (e.g., two agents may disagree on many policies yet broadly agree on what constitutes a more liberal versus more conservative bundle overall). A primary motivation for considering $\mathcal{U}$-robust polarization is to remain agnostic about the exact form of the aggregator $u$. Robust polarization asks: given a class $\mathcal U$ of plausible aggregators, when can the same signal realization polarize aggregate positions for any $u\in\mathcal U$? 

The following result implies a direct mapping between polarization in beliefs and robust polarization in positions (the result follows almost immediately from the definition of a generator for an integral stochastic order in Section \ref{sec:limits}).\footnote{The only subtlety comes from the fact that our definition of polarization requires a strict ordering between the priors and between each prior and its posterior. This maps to the requirement that one of the functions $u\in\mathcal{U}$ leads to a strict inequality in \eqref{eqn:position polarization}. The requirement that $\mathcal{U}$ be closed under addition is mild since if $\succeq$ is generated by $\mathcal{U}$ then it is also generated by the convex cone of $\mathcal{U}$.}

\begin{proposition}\label{prop:positions}
A signal leads to one-shot (respectively, limit) belief polarization according to an integral stochastic order $\succeq$ if and only if it leads to one-shot (respectively, limit) $\mathcal{U}$-robust polarization in positions, where $\mathcal{U}$ is a generator of $\succeq$ that is closed under addition. \end{proposition}

Proposition \ref{prop:positions} converts our order-based possibility/impossibility results into statements about which classes of aggregators admit robust polarization in positions. The coordinatewise order is generated by the class of additively separable increasing functions. Thus, Theorem \ref{thm:cw possible} implies that for this class of aggregators, robust polarization in positions is possible both in the short run and in the limit. The upper orthant order is generated by nonnegative multiplicatively separable increasing functions. This class is not closed under addition, but  by Theorem \ref{thm:uo polarization}, robust polarization for its convex cone is possible only in the short run. In particular, there is less scope for robust polarization once the aggregator exhibits complementarities across dimensions (i.e., is multiplicatively but not additively separable). Finally, the stochastic dominance order is generated by all increasing functions. Thus, Theorem \ref{thm:st impossible} implies that for aggregators that exhibit strong complementarities (i.e., are not even multiplicatively separable), robust polarization in positions is impossible after any signal realization. That is, for any signal, there exists an increasing aggregator with sufficient complementarity under which aggregate positions will not diverge.

These observations highlight a distinction between informational conditions for belief polarization and conditions on the mapping from beliefs to positions. While information about issue alignment can generate belief polarization, strong complementarities in the aggregator may prevent it from translating into polarized positions. Robust polarization in positions is facilitated by information that ties different dimensions together and aggregators that are sufficiently separable across them.

\section{Conclusion}

Whether polarization is rational is widely debated. To make progress, we asked instead which \textit{types} of polarization are consistent with Bayesian updating from public information. For agents learning about a multidimensional state, we delineate the forms of belief divergence that are and are not possible in the short run and in the long-run limit. We find that persistent polarization of all marginal beliefs is consistent with Bayesian updating. On the other hand, we show that more demanding notions of polarization requiring divergence on combinations of issues  are either impossible or short-lived. Evidence of such polarization---for example, beliefs diverging both on the event that the state is at the pointwise minimum and the event that it is at the pointwise maximum---would point to a non-Bayesian explanation. 

Beyond this debate, the Bayesian model isolates the \textit{informational} forces generating polarization from psychological or affective ones. Our characterization result suggests that the empirically observed rise in issue alignment may not simply be a concurrent trend, but a driver of polarization.

\pagebreak

\titleformat{\section}
  {\normalfont\Large\bfseries}
  {Appendix \thesection:}{1em}{}

\begin{appendices}
  \section{Proofs}\label{appendix:proofs}

  This appendix contains the proofs of all results. The proof of Theorem \ref{thm:cw conditions} is in Subsection \ref{subsec:thm cw conditions}, along with the proofs of the lemmas used to prove this theorem. 
  \medskip

\begin{proof}[Proof of Lemma \ref{lemma:limit posterior}]
    Let $\mathcal{G}=\{\Gamma_1,\Gamma_2,\ldots,\Gamma_n\}$ be a partition of $\Theta$ according to the equivalence relation $\ell(\cdot;\theta)\equiv\ell(\cdot;\theta')$. Note that $\hat\theta\in \hat\Gamma\in \mathcal{G}$. For any $\Gamma\in\mathcal{G}$ let $\ell(\cdot;\Gamma)$ be the common likelihood function for states in $\Gamma$. Then, for any $\Gamma\in\mathcal{G}$, if the realizations of the first $t$ signals are $x_1,\ldots,x_t$,

    \[
    Q(t)(\Gamma) = \sum_{\theta\in\Gamma}\frac{P(\theta)\Pi_{\tau=1}^t\ell(x_\tau;\theta)}{\sum_{\theta'\in\Theta}P(\theta')\Pi_{\tau=1}^t\ell(x_\tau;\theta')} = \frac{P(\Gamma)\Pi_{\tau=1}^t\ell(x_\tau;\Gamma)}{\sum_{\Gamma'\in\mathcal{G}}P(\Gamma')\Pi_{\tau=1}^t\ell(x_\tau;\Gamma')}
    \]
    That is, the agent's posteriors over $\mathcal{G}$ are exactly as they would be in a learning problem where the state space is $\mathcal{G}$, the true state is $\hat\Gamma$, and the prior and likelihood function for each state $\Gamma$ are $P(\Gamma)$ and $\ell(\cdot;\Gamma)$, respectively. Note that for this state space it is the case that $\Gamma\neq \Gamma'\Longrightarrow \ell(\cdot;\Gamma)\neq \ell(\cdot;\Gamma')$ (i.e., there exists $x\in\mathcal{X}$ such that $\ell(x;\Gamma)\neq \ell(x;\Gamma')$). Thus, by Doob's Theorem for posterior consistency, $\lim_{t\rightarrow \infty} Q(t)(\hat\Gamma) = 1\text{ a.s.}$\footnote{The original Theorem is \textcite{doob1949application}, here we apply \textcite{miller2018detailed} Theorem 2.4. The theorem has some technical assumptions that are easily satisfied by our finite state space. The key assumption is that the state is identifiable: $\Gamma\neq \Gamma'\Longrightarrow \ell(\cdot;\Gamma)\neq \ell(\cdot;\Gamma')$.}  This implies that for any $\theta\notin\hat\Gamma$, $\lim_{t\rightarrow \infty} Q(t)(\theta) = 0 =Q^*(\theta)$ a.s. Finally, for any $\theta\in\hat\Gamma$, following any sequence $x_1,\ldots,x_t$ of signal realizations, 
 \[
    \frac{Q(t)(\theta)}{Q(t)(\hat\Gamma)} = \frac{P(\theta)\Pi_{\tau=1}^t\ell(x_\tau;\hat\Gamma)}{P(\hat\Gamma)\Pi_{\tau=1}^t\ell(x_\tau;\hat\Gamma)} = \frac{P(\theta)}{P(\hat\Gamma)}
    \]
    So, for any $t$, $Q(t)(\theta) = Q(t)(\hat\Gamma)\frac{P(\theta)}{P(\hat\Gamma)}$, which implies that $\lim_{t\rightarrow \infty} Q(t)(\theta) =\frac{P(\theta)}{P(\hat\Gamma)}=Q^*(\theta)$ a.s.\end{proof}

\begin{proof}[Proof of Theorem \ref{thm:cw possible}] 
Let $\varepsilon\in(0,1)$, set $\Gamma=\{\ubar\theta,\bar\theta\}$ where $\ubar\theta$ is the pointwise minimal state in $\Theta$ and $\bar\theta$ is the pointwise maximal state. Define
\[
P^L(\ubar\theta)=\frac{1}{|\Theta|}+\frac{\varepsilon}{|\Theta|},\, P^L(\bar\theta)=\frac{1}{|\Theta|}-\frac{\varepsilon}{|\Theta|},\text{ and } P^L(\theta)=\frac{1}{|\Theta|}\,\,\forall \theta\notin\Gamma 
\]
Define $P^H$ as a mirror image of $P^L$,
\[
P^H(\ubar\theta)=\frac{1}{|\Theta|}-\frac{\varepsilon}{|\Theta|},\, P^H(\bar\theta)=\frac{1}{|\Theta|}+\frac{\varepsilon}{|\Theta|},\text{ and } P^H(\theta)=\frac{1}{|\Theta|}\,\,\forall \theta\notin\Gamma\]
Let $Q^L$ and $Q^H$ be the agents' posteriors after learning that the identified set is $\Gamma$:
\[
\begin{split}
 Q^L(\ubar\theta)&=\frac{1+\varepsilon}{2},\, Q^L(\bar\theta)=\frac{1-\varepsilon}{2},\text{ and } Q^L(\theta)=0\,\,\forall \theta\notin\Gamma \\
Q^H(\ubar\theta)&=\frac{1-\varepsilon}{2},\, Q^H(\bar\theta)=\frac{1+\varepsilon}{2},\text{ and } Q^H(\theta)=0\,\,\forall \theta\notin\Gamma
\end{split}
\]
Let $i\in[d]$ be a coordinate, and let $n:=|\Theta_i|$ denote the number of elements on this coordinate. Then, the prior and posterior marginals on the $i$ coordinate are:
\[
\begin{split}
   \quad P^L_i &= \left(\frac{1}{n}+ \frac{\varepsilon}{|\Theta|}, \frac{1}{n},\ldots,\frac{1}{n},\frac{1}{n}-\frac{\varepsilon}{|\Theta|} \right),\quad  Q^L_i = \left(\frac{1+\varepsilon}{2}, 0,\ldots,0,\frac{1-\varepsilon}{2} \right)\\
     P^H_i &= \left(\frac{1}{n}-\frac{\varepsilon}{|\Theta|}, \frac{1}{n},\ldots,\frac{1}{n},\frac{1}{n}+\frac{\varepsilon}{|\Theta|} \right),\quad Q^H_i = \left(\frac{1-\varepsilon}{2}, 0,\ldots,0,\frac{1+\varepsilon}{2} \right)
\end{split}
\]
Note that $P^L_i\prec_{st}P^H_i$ for any $\varepsilon>0$, and that both  $Q^L_i\prec_{st} P^L_i$ and $P^H_i\prec_{st} Q^H_i$ hold if and only if $\varepsilon\ge \frac{|\Theta|(n-2)}{n(|\Theta|-2)}$ (recall that we assume $n=|\Theta_i|\ge 2$ and $|\Theta|>2$). The right-hand side of this inequality is strictly less than $1$ since $d\ge 2$ implies $|\Theta|\ge 2n>n$. Hence, there exists $\varepsilon\in(0,1)$ such that coordinatewise polarization occurs. 
\end{proof}

 \begin{proof}[Proof of Lemma \ref{lemma:cw generating}]
    Let $P^L,P^H\in\Delta(\Theta)$ be such that $P^L\preceq_{cw} P^H$. Let $u=\sum _{i=1}^d u_i\in\mathcal{U}$. Since $P^L_i\preceq_{st} P^H_i$ for all $i\in[d]$, we also have $E_{P^L_i}[u_i]\le E_{P^H_i}[u_i]$ for all $i\in[d]$. So, 
    \[
    E^L[u] = \sum_{i=1}^d E_{P^L_i}[u_i] \le \sum_{i=1}^d E_{P^H_i}[u_i] = E^H[u]
    \]

    For the other direction, assume that  $E^L[u]\le E^H[u]$ for all $u\in\mathcal{U}$. Let $j\in[d]$ and let $u_j:\Theta_j\rightarrow \mathbb{R}$ be an increasing function. Define $u_i\equiv 0$ for all $i\neq j$ and  $u:= \sum_{i=1}^d u_i$. Then $u(\theta)=u_j(\theta_j)$ for all $\theta$, so $E^L[u]\le E^H[u]$ implies $E_{P^L_j}[u_j]\le E_{P^H_j}[u_j]$. Since this holds for all $j\in[d]$ and for any increasing $u_j:\Theta_j\rightarrow \mathbb{R}$ we get $P^L_j\preceq_{st} P^H_j$ for all $j\in[d]$, so $P^L\preceq_{cw} P^H$. \end{proof}

\begin{proof}[Proof of Theorem \ref{thm:uo polarization}]
\paragraph{Impossibility of limit $\succeq_{uo}$-polarization: } Assume towards contradiction that limit $\succeq_{uo}$-polarization occurs for an identified set $\Gamma\subset\Theta$. Then $\Gamma\neq \Theta$, so for any $\theta\in\Gamma$ we have $Q^L(\theta)>P^L(\theta)$ and $Q^H(\theta)>P^H(\theta)$. The set $\{\bar\theta\}$ is an upper orthant in $\Theta$.  If $\bar\theta\in\Gamma$ then $Q^L(\bar\theta)>P^L(\bar\theta)$,  a contradiction to $Q^L\prec_{uo} P^L$. Conversely, if $\bar\theta\notin\Gamma$ then  $Q^H(\bar\theta)=0<P^H(\bar\theta)$,  a contradiction to $P^H\prec_{uo} Q^H$. 

\paragraph{Possibility of one-shot $\succeq_{uo}$-polarization:}
    Let $n\in\mathbb{N}$ and $\varepsilon\in(0,1)$. Define priors as follows:
    \[
    \begin{split}
    P^L(\ubar\theta)& = 1-\frac{1}{n}-\frac{1}{n^2}, P^L(\bar\theta) = \frac{1}{n^2}, \text{ and } P^L(\theta) = \frac{1}{n\left(|\Theta|-2\right)} \quad \forall\theta\notin \{\ubar\theta,\bar\theta\}\\
     P^H(\ubar\theta)& =\frac{1}{n^2}, P^H(\bar\theta) =  1-\frac{1}{n}-\frac{1}{n^2}, \text{ and } P^H(\theta) = \frac{1}{n\left(|\Theta|-2\right)}  \quad \forall\theta\notin \{\ubar\theta,\bar\theta\}
    \end{split}
    \]
    Let $x$ be a signal realization with likelihood function given by,
    \[
    \ell(\ubar\theta) =1, \ell(\bar\theta) =1-\varepsilon,\text{ and } \ell(\theta) = 0  \quad \forall\theta\notin \{\ubar\theta,\bar\theta\}
    \]
    Let $Q^L$ and $Q^H$ be the posteriors after observing the realization $x$. We will show upper orthant polarization occurs for sufficiently large $n$. We first show that $P^L\prec_{uo}P^H$. Let $U\subsetneq\Theta$ be an upper orthant. Then $\bar\theta\in U$ and $\ubar\theta\notin U$. So, if $|U|=k$ then, for $n>2$,
    \[
    P^L(U) = \frac{1}{n^2}+\frac{k-1}{n\left(|\Theta|-2\right)} <1- \frac{1}{n}- \frac{1}{n^2}+\frac{k-1}{n\left(|\Theta|-2\right)} = P^H(U)
    \]
    So, $P^L\prec_{uo}P^H$.

 We now show that $Q^L\prec_{uo} P^L$ by proving that for sufficiently large $n$, $Q^L(U)<P^L(U)$ for any upper orthant $U\subsetneq\Theta$. By Lemma \ref{lemma:posterior increase} below, this is equivalent to proving that for sufficiently large $n$, $E^L[\ell | U] < E^L[\ell]$ for any upper orthant $U\subsetneq\Theta$. Let $U\subsetneq\Theta$ be an upper orthant and let $k=|U|$. Then $\bar\theta\in U$ and $\ubar\theta\notin U$ imply, 
 \[
    E^L[\ell|U]  = \frac{1-\varepsilon}{1+\frac{n(k-1)}{|\Theta|-2}}\le 1-\varepsilon
\]
On the other hand, as $n\rightarrow\infty$
\[
    E^L[\ell]=  \left(1-\frac{1}{n}-\frac{1}{n^2}\right)\cdot 1  + \frac{1}{n^2}\cdot (1-\varepsilon) + \frac{1}{n} \cdot 0\rightarrow 1 
\]

Finally, we prove that for sufficiently large $n$, $P^H(U)<Q^H(U)$ for any upper orthant $U\subsetneq\Theta$. Again, let $U\subsetneq\Theta$ be an upper orthant and let $k=|U|$. As $n\rightarrow \infty$, 
\[
 E^H[\ell| U] =\left(1-\frac{1}{n}-\frac{1}{n^2}+\frac{k-1}{n\left(|\Theta|-2\right)}\right)^{-1}\left(\left(1-\frac{1}{n}-\frac{1}{n^2}\right)\cdot(1-\varepsilon) \right)\rightarrow 1-\varepsilon
\]
Since $E^H[\ell]$ is a weighted average of $E^H[\ell| U]$ and $E^H[\ell|\Theta\setminus U]$ to prove that $E^H[\ell|U]>E^H[\ell]$ it suffices to prove that $E^H[\ell|U]>E^H[\ell|\Theta\setminus U]$. This holds for sufficiently large $n$, since as $n\rightarrow\infty$, 
\[
E^H[\ell|\Theta\setminus U] = \frac{1}{1+\frac{n\left(|\Theta|-k-1\right)}{\left(|\Theta|-2\right)}}\rightarrow 0 
\]
For the limit, we used the fact that $|U|=k<|\Theta|-1$. To see this, first note that for any upper orthant $U\subsetneq \Theta$ it must be the case that $\ubar\theta\notin U$. So, if $|U|\ge|\Theta|-1$, then $U=\Theta\setminus\{\ubar\theta\}$. Let $a\in\mathbb{R}^d$ be the \emph{origin} of the upper orthant $U$, i.e., the point for which $U = \{\theta\in\Theta : \theta_1> a_1, \theta_2>a_2,\ldots,\theta_d>a_d\}$. Let $i\in[d]$. Since $U=\Theta\setminus\{\ubar\theta\} $ there exists a state $\theta\in U$ with $\theta_i=\ubar\theta_i$. Thus $a_i\le \theta_i$. Since this holds for all $i\in[d]$ we have $a\le \ubar\theta$, which implies $\ubar\theta\in U$, a contradiction. \end{proof}

\begin{proof}[Proof of Lemma \ref{lemma:posterior increase}]
Let $A\subset \Theta$ be a nonempty set. Since $Q$ is obtained by Bayesian updating with prior $P$ and likelihood $\ell$, we have, 
\[
\begin{split}
&Q(A)>P(A)\iff \frac{\sum_{\theta\in A} P(\theta)\ell(\theta)}{\sum_{\theta\in\Theta}P(\theta)\ell(\theta)}> P(A)\iff \frac{\sum_{\theta\in A} P(\theta)\ell(\theta)}{P(A)} >\sum_{\theta\in\Theta}P(\theta)\ell(\theta)\\
&\iff E_P[\ell |A] > E_P[\ell ]
\end{split}
\]
Here we used the fact that $P$ has full support and that $\ell(\theta)>0$ for at least one $\theta\in \Theta$ to multiply both sides of the inequality by $\frac{\sum_{\theta\in\Theta}P(\theta)\ell(\theta)}{P(A)}$. \end{proof}

\begin{proof}[Proof of Theorem \ref{thm:st impossible}]
Assume towards contradiction that there exists a signal realization $x$ that leads to one-shot $\succeq_{st}$-polarization. We will prove by induction on $n$ that for every upper set $U\subset \Theta$ of size $n$: $P^L(U)=Q^L(U)$ and $P^H(U)=Q^H(U)$. This contradicts the definition of $\succeq_{st}$-polarization, which requires the existence of upper sets $U$ and $U'$ such that $P^L(U)>Q^L(U)$ and $P^H(U')<Q^H(U')$.  
 
 \emph{Base case:} 
 Since $\{\bar\theta\}$ and $\Theta\setminus\{\ubar\theta\}$ are upper sets, the assumption $Q^L\prec_{st}P^L$ implies that $Q^L(\bar\theta)\le P^L(\bar\theta)$ and $Q^L(\ubar\theta)\ge P^L(\ubar\theta)$. Similarly, the assumption $ Q^H\succ_{st} P^H$ implies $Q^H(\bar\theta)\ge P^H(\bar\theta)$ and $Q^H(\ubar\theta)\le P^H(\ubar\theta)$. By Proposition \ref{prop:abstract polarization}, none of these inequalities can be strict, so 
 $Q^L(\bar\theta)=P^L(\bar\theta)$ and $P^H(\bar\theta)=Q^H(\bar\theta)$. Since $\{\bar\theta\}$ is the unique upper set of size $1$ this completes the base case.\footnote{This theorem also holds without the assumption that $\Theta$ is a product space, in which case it may not have a unique maximum. The only change is in the base case, where we would let $\bar\theta$ be some maximal element of $\Theta$ and $\ubar\theta$ be some minimal element of $\Theta$.}

\emph{Inductive step:} Let $1<n\le |\Theta|$, and assume that the inductive hypothesis holds for all $1\le n'< n$. Consider an upper set $U$ of size $n$. Let $\theta_U$ be a minimal element of $U$ and define $A= U\setminus \{\theta_U\}$. Then $A$ is an upper set of size $n-1$. The assumption that $Q^L\prec_{st}P^L$ implies $ P^L(A)+ P^L(\theta_U)=  P^L(U) \ge  Q^L(U) =  Q^L(A)+ Q^L(\theta_U)$. The inductive hypothesis implies $ P^L(A)=  Q^L(A)$, which implies $P^L(\theta_U)\ge  Q^L(\theta_U)$.
By an analogous argument $P^H(\theta_U)\le  Q^H(\theta_U)$. So, as in the base case, Proposition \ref{prop:abstract polarization} implies $P^L(\theta_U)=  Q^L(\theta_U)$ and $P^H(\theta_U)=  Q^H(\theta_U)$. Hence, $ P^L(U)= Q^L(U)$ and  $ P^H(U)= Q^H(U)$, which completes the induction argument. 
\end{proof}

\begin{proof}[Proof of Proposition \ref{prop:abstract polarization}]
Let $\theta^{-}\in\arg\min(\ell)$ and $\theta^{+}\in\arg\max(\ell)$ be states attaining the minimum and maximum likelihood values. Let $E_P[\ell],E_{P'}[\ell]$ denote the expectations of $\ell$ under $P$ and $P'$, respectively. By Lemma \ref{lemma:posterior increase}, for every $\theta\in\Theta$ we have $Q(\theta)>P(\theta)$ if and only if $\ell(\theta)>E_P[\ell]$, and similarly for $Q',P'$. Since $\ell$ is not constant, we have $\ell(\theta^{-})<\ell(\theta^{+})$. Because $P$ and $P'$ have full support, it follows that $\ell(\theta^{-})<\min\{E_P[\ell],E_{P'}[\ell]\}$ and $\ell(\theta^{+})>\max\{E_P[\ell],E_{P'}[\ell]\}$. Thus, for each $\theta\in\{\theta^{-},\theta^{+}\}$ we have $(Q(\theta)-P(\theta))(Q'(\theta)-P'(\theta))>0$. 

To see that beliefs on all other states may move in opposite directions, assume that $\ell(\theta^{-})=a$, $\ell(\theta^{+})=c$ and $\ell(\theta)=b$ for all other $\theta\in\Theta$, where $a<b<c$. It is easy to construct full-support priors with $E_P[\ell]>b>E_{P'}[\ell]$. For such priors, for every state $\theta\in\Theta\setminus\{\theta^{-},\theta^{+}\}$ we have $Q(\theta)-P(\theta)<0<Q'(\theta)-P'(\theta)$.

Finally, assume that there exists $\theta\in\Theta$ for which $Q(\theta)-P(\theta)<0\le Q'(\theta)-P'(\theta)$ (here the first inequality is the strict one w.l.o.g.). Then, $E_P[\ell]> \ell(\theta)\ge E_{P'}[\ell]$. So, there does not exist $\theta'\in\Theta$ at which $Q(\theta')-P(\theta')\ge0\ge Q'(\theta')-P'(\theta')$, as that would imply  $E_P[\ell]\le E_{P'}[\ell]$, a contradiction. 
\end{proof}

 \begin{proof}[Proof of Proposition \ref{prop:positions}] Let $\succeq$ be an integral stochastic order over $\Delta(\Theta)$ generated by a class $\mathcal{U}$ of functions $u:\Theta\rightarrow \mathbb{R}$. We prove the one-shot case; the limit case is analogous.  Assume that one-shot $\succeq$-polarization is possible. Then there exist priors and a signal realization such that $Q^L\prec P^L\prec P^H\prec Q^H$. By definition of a generator, this implies that \eqref{eqn:position polarization} holds for any $u\in\mathcal{U}$. It remains to prove that there exists $u\in\mathcal{U}$ for which all inequalities in \eqref{eqn:position polarization} are strict. Since $Q^L\prec P^L$, there exists $u_1\in\mathcal{U}$ such that $E_{Q^L}[u_1]<E_{P^L}[u_1]$ (otherwise $E_{Q^L}[u]\ge E_{P^L}[u]$ for all $u\in\mathcal{U}$, implying $Q^L\succeq P^L$, a contradiction). Similarly, there exist $u_2,u_3$ in $\mathcal{U}$ such that  $E_{P^L}[u_2]< E_{P^H}[u_2]$ and $E_{P^H}[u_3]<E_{Q^H}[u_3]$. Let $u=u_1+u_2+u_3$. Then, $u\in\mathcal{U}$ (since $\mathcal{U}$ is closed under addition) and is such that all inequalities in \eqref{eqn:position polarization} are strict. For the converse, assume that one-shot $\mathcal{U}$-robust polarization in positions is possible. Then, similar to above, since there exists a signal for which \eqref{eqn:position polarization} holds for any $u\in\mathcal{U}$, we have $ Q^L\preceq P^L\preceq P^H\preceq Q^H$. Since there exists a $u\in\mathcal{U}$ for which \eqref{eqn:position polarization} holds strictly, it must be the case that $ Q^L\prec P^L\prec P^H\prec Q^H$. 
\end{proof}

 \subsection{Proof of Theorem \ref{thm:cw conditions}}\label{subsec:thm cw conditions}
 
\begin{proof}[Proof of Theorem \ref{thm:cw conditions}]
Since $\Theta\subset\mathbb{R}^2$ we can assume $\Theta=\{x_1,\ldots,x_m\}\times\{y_1,\ldots,y_n\}$ for $m,n\ge2$, where $x_1<\ldots<x_m$ and $y_1<\ldots<y_n$. Throughout this proof, let $F^L_i,F^H_i$ denote the cdfs of the marginal distributions of $P^L$ and $P^H$, respectively on the $i$ coordinate. Let $F^L_i|_{\Gamma},F^L_i|_{\Gamma}^C$ denote the cdfs of the marginal distributions of $P^L|_\Gamma$ and $P^L|_{\Gamma^C}$, and similarly define $F^H_i|_{\Gamma},F^H_i|_{\Gamma^C}$.

The proof uses Lemma \ref{lemma:Gamma and complement} to translate polarization into a condition on priors, as well as the following two results on strong stochastic dominance and strong coordinatewise dominance. Proofs of all lemmas stated in the course of the theorem proof appear at the end of this subsection. 

\begin{lemma}\label{lemma:strong st expectation}
    Let $P^L$ and $P^H$ be two univariate distributions with a common finite support $X\subset \mathbb{R}$. Then $P^L\prec^*_{st}P^H$ implies that $E^L[u]<E^H[u]$ for any function $u:X\rightarrow \mathbb{R}$ that is weakly increasing and not constant. 
\end{lemma}

\begin{lemma} \label{lemma:strong cw expectation}
   Let $P^L,P^H\in\Delta(\Theta)$ be such that $P^L\prec^*_{cw} P^H$. Let $u:\Theta\rightarrow\mathbb{R}$, $u=\sum_i u_i$, be a sum of univariate increasing functions and assume that $u$ is not constant over $\Theta$. Then $E^L[u]<E^H[u]$. 
\end{lemma}

\paragraph{The conditions are necessary:}
\begin{enumerate}[(i)]
\item \textit{$\Gamma$ cannot polarize if it is not spanning or if  $\,\Gamma^C$ is not spanning:} Assume that $\Gamma$ is not spanning and assume towards contradiction that coordinatewise polarization occurs. If $\Gamma$ does not obtain the minimal value on one of the coordinates, w.l.o.g the first, then $F^L_1|_{\Gamma}(x_1)=0<F^L_1|_{\Gamma^C}(x_1)$. By Lemma \ref{lemma:Gamma and complement}, this violates $Q^L\prec_{cw} P^L$, a contradiction. Similarly, if $\Gamma$ does not obtain the maximal value on the first coordinate then $F^H_1|_{\Gamma}(x_{m-1})=1<F^H_1|_{\Gamma^C}(x_{m-1})$, a contradiction to $P^H\prec_{cw} Q^H$. Similar arguments show that $\Gamma^C$ must be spanning as well. 

\item \textit{$\Gamma$ cannot polarize if it is not balanced:} Assume w.l.o.g. that $\Gamma$ is biased upward and assume towards contradiction that coordinatewise polarization occurs. Let $\theta^*\in\Theta$  be a state such that $\theta^*\ll \bar\theta$ and such that $\Gamma$ includes all states strictly above $\theta^*$ and none of the states weakly below it. Define the function $u:\Theta\rightarrow \mathbb{R}$ as $u(\theta)=u_1(\theta_1)+u_2(\theta_2)$ where
       \[u_1(x) = \begin{cases}
 0 & x\le\theta^*_1 \\
  1 & x> \theta^*_1
\end{cases}; \quad u_2(y) = \begin{cases}
 0 & y\le\theta^*_2 \\
  1 & y> \theta^*_2
\end{cases} 
\]
Note that $u(\theta)=0$ for all $\theta\le \theta^*$, $u(\theta)=2$ for all $\theta\gg\theta^*$, and $u(\theta)=1$ otherwise. So, the assumption that $\Gamma$ does not contain any state weakly below $\theta^*$ and contains all states strictly above it implies that $E^L[u|\Gamma]>1$ and $E^L[u|\Gamma^C]<1$.\footnote{Recall that these are the conditional expectations according to the prior $P^L$. Since $\theta^*\ll \bar\theta$, there exist states $\theta\in\Theta$ with $\theta\gg \theta^*$. So, the assumption that $P^L$ has full support  implies the strict inequalities.} Since $u$ is a sum of univariate increasing functions, and the coordinatewise order is generated by such functions (Lemma \ref{lemma:cw generating}), this violates $P^L|_{\Gamma}\prec_{cw}P^L|_{\Gamma^C}$, and therefore by Lemma \ref{lemma:Gamma and complement} violates $Q^L\prec_{cw} P^L$, a contradiction.

\item \textit{$\Gamma$ cannot polarize if it is compensatory:} Assume that $\Gamma$ is compensatory and assume towards contradiction that coordinatewise polarization occurs. Let $\theta^*\in\Theta$ be such that $\theta^*\ll\bar\theta$ and such that for every $\theta'\in \Gamma$ it is not the case that $\theta'\le\theta^*$ and it is not the case that $\theta'\gg \theta^*$. Then, for the function $u$ defined in the previous step we have $u(\theta)\equiv 1$ for every $\theta\in\Gamma$. Next, by Lemma \ref{lemma:strong cw expectation}, the strong coordinatewise order implies strict inequality of expectations for nontrivial sums of univariate increasing functions. That is, $P^L\prec^*_{cw} P^H$ and the fact that $u$ is not constant over $\Theta$ imply that  $E^L[u]<E^H[u]$. By Lemma \ref{lemma:Gamma and complement}, $Q^L\prec_{cw} P^L$ implies $P^L|_{\Gamma}\prec_{cw}P^L|_{\Gamma^C}$, which implies 
\[1=E^L[u|\Gamma]\le E^L[u|\Gamma^C]\]
Therefore, 
\[E^L[u]= P^L(\Gamma)E^L[u|\Gamma]+P^L(\Gamma^C)E^L[u|\Gamma^C]\ge 1\]
Similarly, $P^H\prec_{cw} Q^H$ implies   
\[1=E^H[u|\Gamma]\ge E^H[u|\Gamma^C]\]
which implies $E^H[u]\le 1$. So we have $E^L[u]\ge 1\ge E^H[u]$, a contradiction. 
\end{enumerate}

\paragraph{The conditions are sufficient:}
Let $\Gamma\subset\Theta$ be a set that satisfies all three conditions. We will establish the existence of priors $P^L,P^H\in\Delta(\Theta)$ such that strong coordinatewise polarization occurs upon observing $\Gamma$. The proof proceeds in three steps.

\medskip

\emph{\textbf{Step 1}: A method for constructing strongly coordinatewise ordered distributions.}

Recall that a subset of a partially ordered set is an \textit{antichain} if any two distinct elements in the set are incomparable to each other. We will consider antichains that are subsets of $\Theta$, under the order $\le$ defined above. We first define a binary dominance relation over such antichains. 
\begin{definition}\label{def:AC dominance}
    Given two nonempty antichains $A,B\subset \Theta\subset\mathbb{R}^2$ say that $B$ antichain dominates $A$ in $\Theta$, denoted $B>^{\Theta}_{AC} A$ if  the following hold:
    \begin{enumerate}
    \item $A$ and $B$ are disjoint. 
      \item $A$ hits the minimal value  on both coordinates (i.e., there exist $\theta,\theta'\in A$ such that $\theta_1=\ubar\theta_1$, $\theta'_2=\ubar\theta_2$) and $B$ hits the maximal value on both coordinates. 
    \item $A\cup B$ is non-compensatory. 
    \end{enumerate}

    \end{definition}

\begin{lemma}\label{lemma:ac dominance}
 Let $A,B\subset\Theta$ be two nonempty antichains such that $B>^{\Theta}_{AC}A$. Then, there exist distributions $P_A$ with full support on $A$ and $P_B$ with full support on $B$ such that $P_A\prec^*_{cw}P_B$.
\end{lemma}
\begin{remark}
Our original definition of the order $\succeq^*_{cw}$ applied to two distributions fully supported on $\Theta$. Here, the distributions' supports are strictly contained in $\Theta$ but we still take $\succeq^*_{cw}$ to imply the strict marginal inequalities over $\Theta$. That is, $P_A\prec^*_{cw}P_B$ means that $F^A_1(x)>F^B_1(x)$ for each $x\in\Theta_1$ with $x<x_m$ and $F^A_2(y)>F^B_2(y)$ for each $y\in\Theta_2$ with $y<y_n$.     
\end{remark}

\medskip

\emph{\textbf{Step 2}: Applying the method to construct intermediate distributions.} 

We now apply Lemma \ref{lemma:ac dominance} to construct distributions $P^L|_{\min(\Gamma)}, P^L|_{\max(\Gamma^C)},P^H|_{\max(\Gamma)}$ and $ P^H|_{\min(\Gamma^C)}$, each supported on the set indicated in its subscript, such that the following conditions hold:\footnote{We have not yet defined the distributions $P^L$ and $P^H$ but will define them in the next step so that they are consistent with this notation. For example, $P^L|_{\min(\Gamma)}$ will be the conditional distribution of $P^L$ on the set $\min(\Gamma)$.}
\begin{equation}
\label{eqn:conditional dominance}
    \begin{cases}
        P^L|_{\min(\Gamma)} \prec^*_{cw} P^H|_{\max(\Gamma)}\\
        P^L|_{\min(\Gamma)} \prec^*_{cw} P^L|_{\max(\Gamma^C)}\\ 
        P^H|_{\min(\Gamma^C)}  \prec^*_{cw}P^H|_{\max(\Gamma)} 
    \end{cases}
\end{equation}

We will make use of the following result.  
\begin{lemma}\label{lemma:verifying ac}
    Under the conditions of Theorem \ref{thm:cw conditions} all of the following hold: 
    \begin{enumerate}
        \item $\min(\Gamma)<^{\Theta}_{AC}\max(\Gamma)$
        \item $\min(\Gamma)<^{\Theta}_{AC}\max(\Gamma^C)$
        \item $\min(\Gamma^C)<^{\Theta}_{AC}\max(\Gamma)$
    \end{enumerate}
\end{lemma}

We show that we can find distributions that satisfy all the conditions in \eqref{eqn:conditional dominance} by considering four cases depending on whether or not $\Gamma$ contains the extreme states $\ubar\theta,\bar\theta$. 

\textit{Case 1}: $\ubar\theta,\bar\theta\in \Gamma$. In this case, define $P^L|_{\min(\Gamma)}$ and $P^H|_{\max(\Gamma)}$ as Dirac measures on the singleton on which they are supported: $P^L|_{\min(\Gamma)}=D_{\ubar\theta}$ and $P^H|_{\max(\Gamma)}=D_{\bar\theta}$. It is immediate that $P^L|_{\min(\Gamma)}\prec^*_{cw} P^H|_{\max(\Gamma)}$. 

Define $P^L|_{\max(\Gamma^C)}$ and $P^H|_{\min(\Gamma^C)}$ as uniform distributions on their respective supports. By the complement spanning condition, $\max(\Gamma^C)$ obtains the maximal value on each coordinate, which implies that for all $i\in[m-1]$ and $j\in[n-1]$,
 \[
 \begin{split}
   F^L_1|_{\min(\Gamma)}(x_i) &=1 > F^L_1|_{\max(\Gamma^C)}(x_i)    \\
    F^L_2|_{\min(\Gamma)}(y_j) &=1 > F^L_2|_{\max(\Gamma^C)}(y_j)  
 \end{split}
\]
Hence, $P^L|_{\min(\Gamma)}\prec^*_{cw} P^L|_{\max(\Gamma^C)}$. A similar argument shows $ P^H|_{\min(\Gamma^C)}  \prec^*_{cw}P^H|_{\max(\Gamma)}$, so the distributions satisfy \eqref{eqn:conditional dominance}.\footnote{The resulting priors $P^L,P^H$ in this case are going to be those constructed in the proof of Theorem \ref{thm:cw possible}.}

\textit{Case 2}: $\ubar\theta, \bar\theta\in \Gamma^C$. Define $P^H|_{\min(\Gamma^C)}=D_{\ubar\theta}$ and $P^L|_{\max(\Gamma^C)}=D_{\bar\theta}$. By Lemma \ref{lemma:verifying ac} we have $\min(\Gamma)<^{\Theta}_{AC}\max(\Gamma)$. So, by Lemma \ref{lemma:ac dominance}, there exist $P^L|_{\min(\Gamma)}$ fully supported on $\min(\Gamma)$ and $P^H|_{\max(\Gamma)}$ fully supported on $\max(\Gamma)$ such that $P^L|_{\min(\Gamma)}\prec^*_{cw} P^H|_{\max(\Gamma)}$. Similar to Case 1, the spanning condition together with the fact that $P^L|_{\min(\Gamma)},P^H|_{\max(\Gamma)}$ are fully supported implies that $P^L|_{\min(\Gamma)}\prec^*_{cw} P^L|_{\max(\Gamma^C)}$ and $P^H|_{\min(\Gamma^C)}\prec^*_{cw} P^H|_{\max(\Gamma)}$ so \eqref{eqn:conditional dominance} holds. 

\textit{Case 3}: $\ubar\theta\in \Gamma, \bar\theta\in \Gamma^C$. Define $P^L|_{\min(\Gamma)}=D_{\ubar\theta}$ and $P^L|_{\max(\Gamma^C)}=D_{\bar\theta}$. It is immediate that $P^L|_{\min(\Gamma)}\prec^*_{cw}P^L|_{\max(\Gamma^C)} $. By Lemma \ref{lemma:verifying ac} we have $\min(\Gamma^C)<^{\Theta}_{AC}\max(\Gamma)$. So, by Lemma \ref{lemma:ac dominance}, there exist $P^H|_{\min(\Gamma^C)}$ and $P^H|_{\max(\Gamma)}$, each with full support, such that $P^H|_{\min(\Gamma^C)}\prec^*_{cw} P^H|_{\max(\Gamma)}$. Similar to Case 1, the spanning condition and the fact that $P^H|_{\max(\Gamma)}$ has full support implies that $P^L|_{\min(\Gamma)}\prec^*_{cw} P^H|_{\max(\Gamma)}$, so \eqref{eqn:conditional dominance} holds. 

\textit{Case 4}: Follows from Case 3 by swapping $\Gamma\leftrightarrow\Gamma^C$ and $L\leftrightarrow H$.

\bigskip 

\emph{\textbf{Step 3}: Constructing the priors $P^L$ and $P^H$.} 

Let $P^L|_{\min(\Gamma)}, P^L|_{\max(\Gamma^C)},P^H|_{\max(\Gamma)}, P^H|_{\min(\Gamma^C)}$ be distributions that satisfy the conditions in \eqref{eqn:conditional dominance}. Recall that $\Theta=\{x_1,\ldots,x_m\}\times\{y_1,\ldots,y_n\}$. By strong coordinatewise dominance, there exists $\varepsilon>0$ such that for all $i\in[m-1]$ and $j\in[n-1]$,
 \[
 \begin{split}
  F^L_1|_{\min(\Gamma)}(x_i)&>\max\{F^H_1|_{\max(\Gamma)}(x_i),F^L_1|_{\max(\Gamma^C)}(x_i)\}+\varepsilon \\
  F^L_2|_{\min(\Gamma)}(y_j)&>\max\{F^H_2|_{\max(\Gamma)}(y_j),F^L_2|_{\max(\Gamma^C)}(y_j)\}+\varepsilon
 \end{split}
 \]
For any set $A\subset\Theta$, let $U_A$ denote the uniform distribution over $A$. Let $\delta\in(0,1)$ and define\footnote{These distributions are defined as weighted mixtures of their summands. The sets $\Gamma^C\setminus\max(\Gamma^C)$ and $\Gamma^C\setminus\min(\Gamma^C)$ may be empty. This occurs when $\Gamma^C=\min(\Gamma^C)=\max(\Gamma^C)$. In this case, define the corresponding prior as a mixture over the remaining sets, i.e., $P^L = (1-\delta)\left[(1-\delta)P^L|_{\min(\Gamma)}+\delta U_{\Gamma\setminus \min(\Gamma)}\right] +\delta P^L|_{\max(\Gamma^C)}$ and similarly for $P^H$. Both priors are fully supported on $\Theta$ and it is straightforward to verify that all inequalities below hold for this case. The remaining sets in \eqref{eqn:priors} cannot be empty: the min and max sets of any finite set are nonempty, and if it were the case that $\Gamma=\min(\Gamma)=\max(\Gamma)$, then all states in $\Gamma$ would be incomparable to each other, which implies that $\Gamma$ is compensatory.}
 \begin{equation}\label{eqn:priors}
    \begin{split}
          P^L &:= (1-\delta)\left[(1-\delta)P^L|_{\min(\Gamma)}+\delta U_{\Gamma\setminus \min(\Gamma)}\right] +\delta\left[(1-\delta)P^L|_{\max(\Gamma^C)}+\delta U_{\Gamma^C\setminus\max(\Gamma^C)}\right]\\
          P^H &:= (1-\delta)\left[(1-\delta)P^H|_{\max(\Gamma)}+\delta U_{\Gamma\setminus \max(\Gamma)}\right] +\delta\left[(1-\delta)P^H|_{\min(\Gamma^C)}+\delta U_{\Gamma^C\setminus\min(\Gamma^C)}\right]\\ 
    \end{split}
     \end{equation}
     These priors are fully supported on $\Theta$ by construction. We will show that for sufficiently small $\delta$, polarization occurs for the identified set $\Gamma$ and the priors in \eqref{eqn:priors}.

    For any $A\subset\Theta$ let $F^U_i|_{A}$ denote the cdf of the marginal of $U_A$ on the $i$ coordinate. Note that $P^L|_{\Gamma}=(1-\delta)P^L|_{\min(\Gamma)}+\delta U_{\Gamma\setminus \min(\Gamma)}$ and $P^L|_{\Gamma^C}= (1-\delta)P^L|_{\max(\Gamma^C)}+\delta U_{\Gamma^C\setminus\max(\Gamma^C)}$. Thus, for every $i\in[m-1]$ and for $\delta$ sufficiently small such that $(1-\delta)\varepsilon>\delta$ we have
\[
\begin{split}
    F^L_1|_{\Gamma}(x_i)&= (1-\delta)F^L_1|_{\min(\Gamma)}(x_i)+\delta F^U_1|_{\Gamma\setminus \min(\Gamma)}(x_i) \ge (1-\delta)F^L_1|_{\min(\Gamma)}(x_i)\\
    &> (1-\delta) F^L_1|_{\max(\Gamma^C)}(x_i) +(1-\delta)\varepsilon \ge (1-\delta) F^L_1|_{\max(\Gamma^C)}(x_i)+\delta  \cdot 1\\
    &\ge  (1-\delta) F^L_1|_{\max(\Gamma^C)}(x_i)+ \delta F^U_1|_{\Gamma^C\setminus\max(\Gamma^C)}(x_i)  = F^L_1|_{\Gamma^C}(x_i)
\end{split}
\]
Similarly, $F^L_2|_{\Gamma}(y_j)>F^L_2|_{\Gamma^C}(y_j)$ for every $j\in[n-1]$. So $P^L|_{\Gamma}\prec_{cw} P^L|_{\Gamma^C}$. An analogous argument shows that $P^H|_{\Gamma}\succ_{cw}P^H|_{\Gamma^C}$. Thus, by Lemma \ref{lemma:Gamma and complement}, for sufficiently small $\delta$ we have $Q^L\prec_{cw}P^L$ and $P^H\prec_{cw}Q^H$.
 
It remains to prove that $P^L\prec^*_{cw} P^H$. For every $i\in[m-1]$ and for sufficiently small $\delta$,
\[
\begin{split}
       F^L_1(x_i)&\ge (1-\delta)^2F^L_1|_{\min(\Gamma)}(x_i)\\
       &>(1-\delta)^2 F^H_1|_{\max(\Gamma)}(x_i)+(1-\delta)^2\varepsilon> ( 1-\delta)^2 F^H_1|_{\max(\Gamma)}(x_i)+(1-\delta)2\delta +\delta^2 \\
       &\ge (1-\delta)^2F^H_1|_{\max(\Gamma)}(x_i)+(1-\delta)\delta F^U_1|_{\Gamma\setminus \max(\Gamma)}(x_i)\\&+(1-\delta)\delta F^H_1|_{\min(\Gamma^C)}(x_i)
       +\delta^2 F^U_1|_{\Gamma^C\setminus\min(\Gamma^C)}(x_i) = F^H_1(x_i)  
\end{split}
\] 
Here we take $\delta$ such that $(1-\delta)^2\varepsilon > (1-\delta)2\delta +\delta^2$. This inequality is satisfied for sufficiently small $\delta>0$ because as $\delta\rightarrow 0$ the right-hand-side converges to $0$ while the left-hand-side converges to $\varepsilon>0$. A similar argument shows that $F^L_2(y_j)>F^H_2(y_j)$ for every $j\in[n-1]$.  So $P^L\prec^*_{cw} P^H$.
\end{proof}

\begin{proof}[Proof of Lemma \ref{lemma:Gamma and complement}]
    For  any $\theta\in\Gamma$: $Q(\theta)-P(\theta) = \frac{P(\theta)}{P(\Gamma)} -P(\theta) = \frac{P(\theta)}{P(\Gamma)}P(\Gamma^C)$. So, for any coordinate $i\in[d]$ and $x\in\mathbb{R}$,   
 \[
\begin{split}
&Q(\theta_i\le x)- P(\theta_i\le x) = \sum_{\theta\in\Gamma;\theta_i\le x} \left[P(\theta)+\frac{P(\theta)}{P(\Gamma)}P(\Gamma^C)\right] - \sum_{\theta\in\Theta;\theta_i\le x} P(\theta) = \\ 
&       \frac{P(\Gamma^C)}{P(\Gamma)}\sum_{\theta\in\Gamma;\theta_i\le x}P(\theta)-\sum_{\theta\in\Gamma^C;\theta_i\le x} P(\theta)=
\frac{P(\Gamma^C)}{P(\Gamma)}P(\Gamma\cap\{\theta_i\le x\}) -P(\Gamma^C\cap\{\theta_i\le x\})
\end{split}
\]
So,
\[
\begin{split}
 &Q\preceq_{cw} P \iff Q(\theta_i\le x)\ge P(\theta_i\le x)\quad \forall i\in[d],x\in\mathbb{R}\iff\\ 
&\frac{P(\Gamma\cap\{\theta_i\le x\})}{P(\Gamma)} \ge \frac{P(\Gamma^C\cap\{\theta_i\le x\})}{P(\Gamma^C)}\quad\forall i\in[d],x\in\mathbb{R}\iff P|_{\Gamma} \preceq_{cw}P|_{\Gamma^C}   
\end{split}
\]
\end{proof}

\begin{proof}[Proof of Lemma \ref{lemma:strong st expectation}]
Let $F^L$ and $F^H$ be the corresponding cdfs. Denote $X=\{x_1,\ldots,x_n\}$, where $x_1<\ldots<x_n$. Using summation by parts,   
   \[
   \begin{split}
        E^H[u]-E^L[u] &= \sum_{i=1}^n u(x_i)\left(P^H(x_i)-P^L(x_i)\right)\\ & = u(x_n)\left(F^H(x_n)-F^L(x_n)\right)-  \sum_{i=1}^{n-1}\left(F^H(x_i)-F^L(x_i)\right)\left(u(x_{i+1})-u(x_i)\right)  \\
        & = \sum_{i=1}^{n-1}\left(F^L(x_i)-F^H(x_i)\right)\left(u(x_{i+1})-u(x_i)\right)
   \end{split}
   \]
  Since $P^L\prec^*_{st} P^H$, we have $F^L(x_i)-F^H(x_i)>0$ for every $i\in[n-1]$. This, together with the fact that $u$ is weakly increasing and not constant implies that the summands in the final expression are weakly greater than 0 for all $i\in[n-1]$ and strictly greater than 0 for some $i\in[n-1]$, hence $E^H[u]-E^L[u]>0$.
\end{proof}

\begin{proof}[Proof of Lemma \ref{lemma:strong cw expectation}]
    Since $u$ is not constant, one of the functions $u_i$ must not be constant. Thus, since $P^L_i\prec^*_{st} P^H_i$ for all $i\in[d]$, Lemma \ref{lemma:strong st expectation} implies that $E_{P^L_i}[u_i]\le E_{P^H_i}[u_i]$ for all $i\in[d]$ and $E_{P^L_j}[u_j]< E_{P^H_j}[u_j]$ for at least one $j\in[d]$. So, 
    \[
    E^L[u] = \sum_{i=1}^d E_{P^L_i}[u_i] < \sum_{i=1}^d E_{P^H_i}[u_i] = E^H[u]
    \]
    \end{proof}
 
\begin{proof}[Proof of Lemma \ref{lemma:ac dominance}] For  the proof it will be easier to denote $P^A$ and $P^B$ rather than $P_A,P_B$ so that $P^A_i$ denotes the marginal of $P^A$ on the $i$ coordinate, $F^A_i$ denotes the cdf of $P^A_i$, and likewise for $B$.

We prove by induction on $s=|\Theta|$. Note that for $\Theta$ to be a two-dimensional product set its cardinality $s$ must be a composite number so the base case is $s=4$ and the induction is over all composite numbers $s$.   

\emph{Base case:} Assume $\Theta=\{x_1,x_2\}\times\{y_1,y_2\}$. Let $A,B\subset\Theta$ be two nonempty antichains such that $B>^{\Theta}_{AC}A$. Note that, apart from the two-element antichain $\{(x_2,y_1),(x_1,y_2)\}$, all other antichains in $\Theta$ are singletons. Hence, at least one of the disjoint antichains  $A,B$ is a singleton. Assume w.l.o.g. that this is $B$. Since $B$ hits the maximal value on both coordinates, it must be the case that $B=\{(x_2,y_2)\}$. Let $P^B$ be the Dirac measure on $B$ and $P^A$ be any full support measure on $A$. Then 
$F^A_1(x_1)>0=F^B_1(x_1)$ and $F^A_2(y_1)>0= F^B_2(y_1)$ so $P^A\prec^*_{cw}P^B$. 

\emph{Inductive step:} Assume the induction hypothesis holds for all composite numbers $s'<s$ and consider a state space $\Theta$ with $|\Theta|=s$.  Let $A,B\subset\Theta$ be two nonempty antichains such that $B>^{\Theta}_{AC}A$. If one of these is a singleton then the proof is straightforward by an argument similar to the base case. Assume that $|A|,|B|\ge 2$. We can write $A=\{\delta^1,\ldots,\delta^k\}, B=\{\theta^1,\ldots,\theta^{\ell}\}$, where both sets are ordered according to the first coordinate from the least to greatest. That is, $\delta^1_1<\delta^2_1<\ldots<\delta^k_1$, and similarly for $B$ (superscripts index the elements of $A,B$, subscripts denote coordinates). Since $A$ and $B$ are antichains it must be the case that the order on the second coordinate is opposite, i.e., in $A$, $\delta^1_2>\delta^2_2,\ldots>\delta^k_2$ and similarly in $B$. Since $A$ hits the minimal value on both coordinates and $B$ hits the maximal value on both coordinates it follows that $\delta^1_1$ is the minimum value on the first coordinate in $\Theta$, whereas $\theta^1_2$ is the maximal value on the second coordinate in $\Theta$. In particular, we have $\delta^1_1\le \theta^1_1$ and $\delta^1_2\le\theta^1_2,$ so $\delta^1\le \theta^1$. 

Next, we claim that one of the following holds 
\begin{equation}
\label{eq:disjunctive_conditions}
\left\{
\begin{array}{l}
    \text{(i)  }\quad\delta^1\le \theta^1,\theta^2;\text{ or } \\

    \text{(ii)  }\quad \theta^1\ge \delta^1,\delta^2
\end{array}
\right.
\end{equation}
Assume towards contradiction that this is not the case. Then, since $\delta^1$ is minimal on the first coordinate (obtains the minimal value in $\Theta_1$), and $\theta^1$ is maximal on the second coordinate it must be the case that $\delta^1_2>\theta^2_2$ and $\theta^1_1<\delta^2_1$. Let $\tau = (\theta^1_1,\max\{\theta^2_2,\delta^2_2\})\in\Theta$. Note that $\tau\ll\bar\theta$.

We claim that for every $\theta\in A\cup B$ neither $\theta\le \tau$ nor $\theta\gg\tau$ hold, so that $A\cup B$ is compensatory in contradiction to the third condition in Definition \ref{def:AC dominance}. To see this, note that $\delta^1_1\le \theta^1_1=\tau_1$ and $\delta^1_2>\max\{\theta^2_2,\delta^2_2\}=\tau_2$. So $\delta^1$ is neither weakly below $\tau$ nor strictly above it. Since $A$ is an antichain ordered according to the first coordinate, for any $2\le r\le k$ we have $\delta^r_1\ge\delta^2_1>\theta^1_1$ and $\delta^r_2\le \delta^2_2\le\max\{\theta^2_2,\delta^2_2\}$,  so that all states in $A$ are neither weakly below $\tau$ nor strictly above it. 

Likewise, $\theta^1_1=\tau_1$ and, since $\theta^1$ is maximal on the second coordinate, $\theta^1_2>\tau_2$.   So $\theta^1$ is neither weakly below $\tau$ nor strictly above it. Since $B$ is an antichain ordered according to the first coordinate, for any $2\le r\le \ell$ we have $\theta^r_1>\theta^1_1=\tau_1$ and $\theta^r_2\le \theta^2_2\le \tau_2$, so that all states in $B$ are neither weakly below $\tau$ nor strictly above it. Thus, $A\cup B$ is compensatory, a contradiction. 

We have established that one of the conditions in \eqref{eq:disjunctive_conditions} holds. Assume w.l.o.g. that (i) holds (the proof for (ii) is symmetric).  Let $B'=B\setminus\{\theta^1\}$ and let  $\Theta'=\{\theta\in\Theta|\theta_2\le\theta^2_2\}$. We claim that $B'>^{\Theta'}_{AC}A$. Since $\delta^1_2\le \theta^2_2$ we have $A\subset\Theta'$ and therefore $A$ and $B'$ are disjoint antichains in $\Theta'$, so the first condition in Definition \ref{def:AC dominance} is satisfied. For the second condition, note that $A$ hits the minimal values in $\Theta'$ on both coordinates because these are the same minimal values as in $\Theta$. And $B'$ hits the maximal values in $\Theta'$ on both coordinates, because the maximal value in $\Theta'$ on the first coordinate is the same as in $\Theta$ (and by assumption this is the value of $\theta^{\ell}_1$) and the maximal value on the second coordinate is $\theta^2_2$.

To see that the third condition in Definition \ref{def:AC dominance} holds,  assume towards contradiction that $A\cup B'$ is compensatory (in $\Theta'$). Let $\bar\theta'=(\bar\theta_1,\theta^2_2)$ denote the maximal value in $\Theta'$. Let $\theta^*\in\Theta'$ be such that $\theta^*\ll\bar\theta'$ and such that every state in $ A\cup B'$ is neither weakly below $\theta^*$ nor strictly above it. This implies that $\theta^2_1\le \theta^*_1$ (otherwise $\theta^2\gg \theta^*$), which in turn implies that $\theta^1_1< \theta^*_1$. But we also have $\theta^1_2>\theta^*_2$, so $\theta^1$ is neither weakly below $\theta^*$ nor strictly above it. But then every state in $ A\cup B$ is neither weakly below $\theta^*$ nor strictly above it and $\theta^*\ll \bar\theta$, so $A\cup B$ is compensatory in $\Theta$ in contradiction to $B>^{\Theta}_{AC}A$. 

We have established that $B'>^{\Theta'}_{AC}A$. Since $|\Theta'|<s$, by the induction hypothesis there exist distributions $P^A,P^{B'}$ supported on $A$ and $B'$, respectively, such that $P^A\prec^*_{cw}P^{B'}$ on $\Theta'$. Recall that $\ubar\theta$ is the minimal state in $\Theta$ and $\bar\theta'=(\bar\theta_1,\theta^2_2)$ is the maximal state in $\Theta'$. So $P^A\prec^*_{cw}P^{B'}$ implies that there exists $\varepsilon>0$ such that for every coordinate $i$ and for every $t\in\Theta'_i$ with $t<\max(\Theta'_i)$ we have $F^A_i(t)>F^{B'}_i(t)+2\varepsilon$. Let $D_{\theta^1}$ be the Dirac measure on $\theta^1$ and define the mixture distribution $P^B =(1-\varepsilon)P^{B'}+\varepsilon D_{\theta^1}$. Then, for any  $x\in \Theta'_1=\Theta_1$, with $x<\bar\theta'_1=\bar\theta_1$ we have  
 \[F^B_1(x)\le (1-\varepsilon)F^{B'}_1(x)+\varepsilon < F^{B'}_1(x)+2\varepsilon <F^A_1(x)\]
So, $P^A_1\prec^*_{st} P^B_1$ in $\Theta_1$.

Moving on to the second coordinate, fix any $y\in\Theta_2$ with $y<\bar\theta_2$. Since $\theta^1_2=\bar\theta_2$, the Dirac mass on $\theta^1$ assigns no probability to $\{\theta_2\le y\}$, and therefore
\[
F^B_2(y)=(1-\varepsilon)F^{B'}_2(y).
\]
If $F^{B'}_2(y)>0$, then $(1-\varepsilon)F^{B'}_2(y)<F^{B'}_2(y)\le F^A_2(y)$. If instead $F^{B'}_2(y)=0$, then $F^B_2(y)=0$ while $F^A_2(y)>0$ because $P^A$ has full support on $A$ and $A$ hits the minimal value on the second coordinate, so $P^A_2(\ubar\theta_2)>0$. In either case, we have $F^B_2(y)<F^A_2(y)$.

Hence, $P^A_2\prec^*_{st} P^B_2$ on $\Theta_2$ and therefore $P^A\prec^*_{cw} P^B$ on $\Theta$, which completes the induction argument.  
\end{proof}

\begin{proof}[Proof of Lemma \ref{lemma:verifying ac}]
    We prove these in order, by showing that each pair of sets satisfies the three conditions in Definition \ref{def:AC dominance}.
    
    \medskip 
    
\noindent\emph{\textbf{Proving that $\min(\Gamma) <^{\Theta}_{AC}\max(\Gamma)$}}:
\begin{enumerate}
    \item \emph{$\min(\Gamma)$ and $\max(\Gamma)$ are disjoint}: Assume towards contradiction that there exists $\theta\in \min(\Gamma)\cap\max(\Gamma)$. There must be one coordinate in which $\theta$ does not obtain the maximal value and one coordinate in which $\theta$ does not obtain the minimal value. Otherwise, if for example $\theta=\bar\theta$ then since $\theta\in\min(\Gamma)$ it must be the case that $\Gamma=\{\theta\}$ but then $\Gamma$ is not spanning. Assume w.l.o.g. that $\theta$ is not minimal on the first coordinate and is not maximal on the second coordinate. 
    
    Recall that $\Theta_1= \{x_1,\ldots,x_m\}$ with $x_1<x_2<\ldots<x_m$. The assumption that $\theta$ is not minimal on the first coordinate implies $\theta_1= x_j$ for some $1<j\le m$. Define $\theta^* = (x_{j-1}, \theta_2)$. Then $\theta^*\ll \bar\theta$.   We will show that for every $\theta'\in \Gamma$ it is not the case that $\theta'\le\theta^*$ and it is not the case that $\theta'\gg \theta^*$ so $\Gamma$ is compensatory in contradiction to the third condition in Theorem \ref{thm:cw conditions}. 

    If there exists $\theta'\in\Gamma$ with $\theta'\le \theta^*$ then $\theta'\le \theta$ and $\theta'\neq \theta$ in contradiction to $\theta\in\min(\Gamma)$. If there exists $\theta'\in \Gamma$ with $\theta'\gg\theta^*$ then $\theta'\ge\theta$ and $\theta'\neq \theta$ in contradiction to $\theta\in\max(\Gamma)$. 
    
   \item \emph{$\min(\Gamma)$ hits the minimal value  on both coordinates and $\max(\Gamma)$ hits the maximal value on both coordinates:} Follows immediately from the assumption that $\Gamma$ is spanning. 
   \item \emph{$\min(\Gamma)\cup \max(\Gamma)$ is non-compensatory:} Assume towards contradiction that there exists $\theta\in\Theta$, with $\theta\ll\bar\theta$, such that for every $\theta'\in \min(\Gamma)\cup \max(\Gamma)$ it is not the case that $\theta'\le\theta$ and it is not the case that $\theta'\gg \theta$. Then there does not exist $\theta'\in\Gamma$ with $\theta'\gg \theta$ since that would imply the existence of $\tilde\theta\in\max(\Gamma)$ such that $\tilde\theta\ge \theta'\gg\theta$, a contradiction. Similarly, there does not exist $\theta'\in\Gamma$ with $\theta'\le \theta$ as that would imply the existence of  $\tilde\theta\in\min(\Gamma)$ such that $\tilde\theta\le\theta'\le \theta$. So $\Gamma$ is compensatory, a contradiction.  
    
\end{enumerate}

   \medskip 
    
\noindent\emph{\textbf{Proving that $\min(\Gamma)<^{\Theta}_{AC}\max(\Gamma^C)$}}:
\begin{enumerate}
    \item \emph{$\min(\Gamma)$ and $\max(\Gamma^C)$ are disjoint}: immediate.
    \item \emph{$\min(\Gamma)$ hits the minimal value  on both coordinates and $\max(\Gamma^C)$ hits the maximal value on both coordinates:} Follows immediately from the assumption that $\Gamma$ and $\Gamma^C$ are spanning.  
     \item \emph{$\min(\Gamma)\cup \max(\Gamma^C)$ is non-compensatory:} Assume towards contradiction that there exists $\theta\in\Theta$, with $\theta\ll\bar\theta$, such that for every $\theta'\in \min(\Gamma)\cup \max(\Gamma^C)$ it is not the case that $\theta'\le\theta$ and it is not the case that $\theta'\gg \theta$. Then $\Gamma$ includes all states strictly above $\theta$, because if $\theta'\gg\theta$ and $\theta'\in\Gamma^C$ then there exists some $\tilde\theta\in\max(\Gamma^C)$ such that $\tilde\theta\ge\theta'\gg\theta$, a contradiction. Similarly, $\Gamma^C$ includes all states weakly below $\theta$, as otherwise there exists $\tilde\theta\in\min(\Gamma)$ with $\tilde\theta\le \theta$. Thus, $\Gamma$ is biased upward, in contradiction to the second condition in Theorem \ref{thm:cw conditions}.
    \end{enumerate} 

\medskip 
    
\noindent\emph{\textbf{Proving that $\min(\Gamma^C)<^{\Theta}_{AC}\max(\Gamma)$}}: 
The arguments for the first two conditions of Definition \ref{def:AC dominance} are exactly as those in the proof that $\min(\Gamma)<^{\Theta}_{AC}\max(\Gamma^C)$. For the third condition, assume towards contradiction that there exists $\theta\in\Theta$, with $\theta\gg\ubar\theta$, such that for every $\theta'\in \min(\Gamma^C)\cup \max(\Gamma)$ it is not the case that $\theta'\ge\theta$ and it is not the case that $\theta'\ll \theta$.\footnote{Here we are using the equivalent definition for a set to be compensatory that is stated in Footnote \ref{fn:equivalent compensatory}. } Then, $\Gamma$ includes all states strictly below $\theta$, because if there exists $\theta'\in\Gamma^C$ with $\theta'\ll\theta$ then there exists $\tilde\theta\in\min(\Gamma^C)$ such that $\tilde\theta\le \theta'\ll\theta$, a contradiction. Similarly, $\Gamma^C$ includes all states weakly above $\theta$. Thus, $\Gamma$ is biased downward in contradiction to the second condition in Theorem \ref{thm:cw conditions}.
\end{proof}

\section{Probability of Polarization}\label{appendix:probability}
The following example demonstrates that limit coordinatewise polarization is not only possible but can also be very probable. However, there is a tradeoff between the probability of polarization and its magnitude.

\begin{example}\label{example:probability}
    Let $\delta\in(0,1)$. Consider the $2\times 2$ state space $\Theta=\{x_1,x_2\}\times\{y_1,y_2\}$ with the priors given in Figure \ref{fig:example 2 priors}. 
\begin{figure}[H]
\centering
\begin{minipage}[b]{0.48\textwidth}
\centering
\begin{tikzpicture}[scale=1.75]
  \draw[->, thick] (-0.3,0) -- (2.2,0) node[right] {$x$};
  \draw[->, thick] (0,-0.3) -- (0,2.2) node[above] {$y$};
  \draw[thick] (0,0) rectangle (2,2);
  \draw[thick] (1,0) -- (1,2);
  \draw[thick] (0,1) -- (2,1);
  \draw (0.5, -0.03) -- (0.5, 0.03) node[below] {$x_1$};
  \draw (1.5, -0.03) -- (1.5, 0.03) node[below] {$x_2$};
  \draw (-0.03, 0.5) -- (0.03, 0.5) node[left] {$y_1$};
  \draw (-0.03, 1.5) -- (0.03, 1.5) node[left] {$y_2$};
  \node at (0.5,0.5) {$\mathbf{(1-\delta)\frac{3}{4}}$};
  \node at (1.5,1.5) {$\mathbf{(1-\delta)\frac{1}{4}}$};
  \node at (1.5,0.5) {$\mathbf{\delta\frac{1}{2}}$};
  \node at (0.5,1.5) {$\mathbf{\delta\frac{1}{2}}$};
  \node[anchor=north east, font=\large\bfseries] at (0,0) {\(\mathbf{P^L}\)};
\end{tikzpicture}
\end{minipage}\hfill
\begin{minipage}[b]{0.48\textwidth}
\centering
\begin{tikzpicture}[scale=1.75]
  \draw[->, thick] (-0.3,0) -- (2.2,0) node[right] {$x$};
  \draw[->, thick] (0,-0.3) -- (0,2.2) node[above] {$y$};
  \draw[thick] (0,0) rectangle (2,2);
  \draw[thick] (1,0) -- (1,2);
  \draw[thick] (0,1) -- (2,1);
  \draw (0.5, -0.03) -- (0.5, 0.03) node[below] {$x_1$};
  \draw (1.5, -0.03) -- (1.5, 0.03) node[below] {$x_2$};
  \draw (-0.03, 0.5) -- (0.03, 0.5) node[left] {$y_1$};
  \draw (-0.03, 1.5) -- (0.03, 1.5) node[left] {$y_2$};
  \node at (0.5,0.5) {$\mathbf{(1-\delta)\frac{1}{4}}$};
  \node at (1.5,0.5) {$\mathbf{\delta\frac{1}{2}}$};
  \node at (0.5,1.5) {$\mathbf{\delta\frac{1}{2}}$};
  \node at (1.5,1.5) {$\mathbf{(1-\delta)\frac{3}{4}}$};
  \node[anchor=north east, font=\large\bfseries] at (0,0) {\(\mathbf{P^H}\)};
\end{tikzpicture}
\end{minipage}
\caption{Priors for Example \ref{example:probability}}
\label{fig:example 2 priors}
\end{figure}
The posteriors upon observing a partitional signal with identified set  $\Gamma=\{(x_1,y_1), (x_2,y_2)\}$ are given in Figure \ref{fig:example 2 posteriors}. 
\begin{figure}[htbp]
\centering
\begin{minipage}[b]{0.48\textwidth}
\centering
\begin{tikzpicture}[scale=1.75]
  \draw[->, thick] (-0.3,0) -- (2.2,0) node[right] {$x$};
  \draw[->, thick] (0,-0.3) -- (0,2.2) node[above] {$y$};
  \draw[thick] (0,0) rectangle (2,2);
  \draw[thick] (1,0) -- (1,2);
  \draw[thick] (0,1) -- (2,1);
  \draw (0.5, -0.03) -- (0.5, 0.03) node[below] {$x_1$};
  \draw (1.5, -0.03) -- (1.5, 0.03) node[below] {$x_2$};
  \draw (-0.03, 0.5) -- (0.03, 0.5) node[left] {$y_1$};
  \draw (-0.03, 1.5) -- (0.03, 1.5) node[left] {$y_2$};
  \node at (0.5,0.5) {\textbf{3/4}};
  \node at (1.5,0.5) {\textbf{0}};
  \node at (0.5,1.5) {\textbf{0}};
  \node at (1.5,1.5) {\textbf{1/4}};
  \node[anchor=north east, font=\large\bfseries] at (0,0) {\(\mathbf{Q^L}\)};
\end{tikzpicture}
\end{minipage}\hfill
\begin{minipage}[b]{0.48\textwidth}
\centering
\begin{tikzpicture}[scale=1.75]
  \draw[->, thick] (-0.3,0) -- (2.2,0) node[right] {$x$};
  \draw[->, thick] (0,-0.3) -- (0,2.2) node[above] {$y$};
  \draw[thick] (0,0) rectangle (2,2);
  \draw[thick] (1,0) -- (1,2);
  \draw[thick] (0,1) -- (2,1);
  \draw (0.5, -0.03) -- (0.5, 0.03) node[below] {$x_1$};
  \draw (1.5, -0.03) -- (1.5, 0.03) node[below] {$x_2$};
  \draw (-0.03, 0.5) -- (0.03, 0.5) node[left] {$y_1$};
  \draw (-0.03, 1.5) -- (0.03, 1.5) node[left] {$y_2$};
  \node at (0.5,0.5) {\textbf{1/4}};
  \node at (1.5,0.5) {\textbf{0}};
  \node at (0.5,1.5) {\textbf{0}};
  \node at (1.5,1.5) {\textbf{3/4}};
  \node[anchor=north east, font=\large\bfseries] at (0,0) {\(\mathbf{Q^H}\)};
\end{tikzpicture}
\end{minipage}
\caption{Posteriors for Example \ref{example:probability}}
\label{fig:example 2 posteriors}
\end{figure}

In this symmetric example, a natural measure for the \emph{magnitude} of polarization as a function of $\delta$ is the difference between the post- and pre-signal gaps in marginal beliefs:
\[
M(\delta) = Q^L_1(x_1)-Q^H_1(x_1)-(P^L_1(x_1)-P^H_1(x_1)) = Q^L_2(y_1)-Q^H_2(y_1)-(P^L_2(y_1)-P^H_2(y_1))  
\]
A simple computation shows $M(\delta)= \frac{\delta}{2}$, so the magnitude of polarization is linearly increasing in $\delta$. 

Polarization only occurs on $\Gamma$, to which both agents assign prior probability $1-\delta$. This holds regardless of how the partitional signal is defined outside $\Gamma$: it may fully reveal the state or only reveal that it is not in $\Gamma$, and neither case polarizes. Thus, there is a direct tradeoff between the probability of polarization and its magnitude. When $\delta$ is close to $0$, polarization occurs with probability close to $1$ according to both agents' subjective beliefs, but its magnitude is small. As $\delta$ increases, polarization becomes less likely, but its magnitude increases.

\end{example}

\printbibliography

\end{appendices}

\end{document}